\documentclass[aps,superscriptaddress,11pt]{revtex4-2}
\usepackage{amsmath, amssymb, amsthm, graphicx, tikz, booktabs, bm}
\usepackage{caption}
\usepackage{pgfplots}
\usepackage{tikz-3dplot}
\usepackage{tikzscale}
\usepackage{times}
\usepackage{hyperref}
\usetikzlibrary{positioning, arrows.meta, decorations.pathmorphing, calc}
\pgfplotsset{compat=1.18}

\renewcommand{\thesection}{\Roman{section}}

\makeatletter
\def\@hangfrom@section#1#2#3{\@hangfrom{{\large #1} #2}{\large #3}}
\def\@hangfroms@section#1#2{#1#2}
\makeatother

\def\b{\beta}

\newcommand{\EMD}{\mathrm{EMD}}

\newcommand{\NMR}{\mathrm{NMR}}

\newcommand{\Rn}{R^{(n)}_t}
\newcommand{\Rinf}{R^{(\infty)}_t}
\newcommand{\Tn}{T^{(n)}_t}
\newcommand{\Tinf}{T^{(\infty)}_t}

\newtheorem{theorem}{Theorem}[section]
\newtheorem{proposition}[theorem]{Proposition}
\newtheorem{lemma}[theorem]{Lemma}
\newtheorem{corollary}[theorem]{Corollary}
\theoremstyle{definition}

\newtheorem{remark}[theorem]{Remark}
\theoremstyle{plain}

\begin{document}

\title{Weight geometry governs functional memory in complex systems}

\author{Elka\"ioum M.~Moutuou}
\email{elkaioum.moutuou@concordia.ca}
\affiliation{Department of Electrical and Computer Engineering, Concordia University, Montreal, QC, H3G 1M8}

\author{Habib Benali}
\affiliation{Department of Electrical and Computer Engineering, Concordia University, Montreal, QC, H3G 1M8}

\begin{abstract}
Complex systems, from gene regulatory networks to neural circuits and
ecological food webs, exhibit rich functional behaviour that topology
alone does not capture. Yet functional complexity remains difficult to
quantify independently of structural organisation. Here we introduce a
thermodynamic framework in which functional complexity is characterised through
the hierarchical organisation of functional memory, quantifying how the
influence of past interactions is distributed and progressively compressed
across scales. Across thirty-four empirical networks spanning biological,
ecological, social, technological, and biophysical systems and several
orders of magnitude in size and density, real interaction strengths
organise functional memory at greater hierarchical depth than random
weight assignment on the same topology in every domain studied. The
framework further reveals that functional memory occupies a remarkably
low-dimensional space, collapsing onto four recurrent dynamical
organisations. Comparisons with null models that selectively perturb
weighted transport geometry, mesoscale wiring, and directionality show
that these structural ingredients play distinct roles: weighted transport
geometry systematically governs memory depth, whereas mesoscale wiring  
organises memory across scales and directionality modulates the response of 
the cascade to structural perturbation. The same comparisons provide an operational 
criterion for determining whether network weights encode functionally meaningful
structure beyond topology. These results establish weighted transport
geometry as a primary organiser of functional memory and provide a quantitative framework for 
studying functional complexity in directed weighted networks.
\end{abstract}

\maketitle

\section{Introduction}
What principle governs how living, social, and engineered systems store the
memory of their past interactions? Network science has traditionally sought
the answer in structural properties of the wiring diagram, including degree
heterogeneity, modularity, and spectral
organisation~\cite{Simon1991,Newman2003,Bianconi2023}. This topological
perspective has been productive, yet it implicitly assumes that the functional
complexity of a system is largely determined by which connections exist. Even minimal 
systems illustrate the distinction between structure and the
dynamics of historical integration. Networks with similar static structural
properties can nevertheless compress causal history at markedly different
rates, revealing differences invisible to conventional structural
observables (Supplementary Section 1). Here we show that this assumption fails in one precise 
and empirically falsifiable respect: the geometry of interaction strengths consistently
shapes how deeply past interactions remain integrated into present dynamics.

In directed networks, information propagates locally along directed paths
and integrates contributions from progressively earlier interactions. The depth over 
which past interactions continue to influence a node's present state is finite and 
constrained by both structure and physical cost.
Structural observables characterise which paths exist, but they do not resolve
how path contributions are weighted across temporal depth, how historical
influence is suppressed as the system evolves, or how this suppression is
distributed across scales. The organisation of functional memory is 
encoded precisely in these quantities, and they lie beyond the reach of topological
analysis alone.

Across thirty-four empirical networks spanning gene regulation, protein
interaction, neural connectomes, ecological webs, social and information
systems, and transportation infrastructure, randomising edge weights while
preserving binary topology reduces functional memory depth. The effect
holds across all domains and over three orders of magnitude in network
size. The mechanism follows directly from the structure of directed walks:
functional memory is sustained by long weighted paths, and real systems
correlate strong interactions along coherent directed pathways, producing
multiplicative amplification of deep contributions that random weight
configuration suppresses. Topology determines which paths exist, whereas
the geometry of directed transport determines how strongly distant causal
histories remain integrated.

\begin{figure*}[th!]
 \centering
 \includegraphics[width=.8\textwidth]{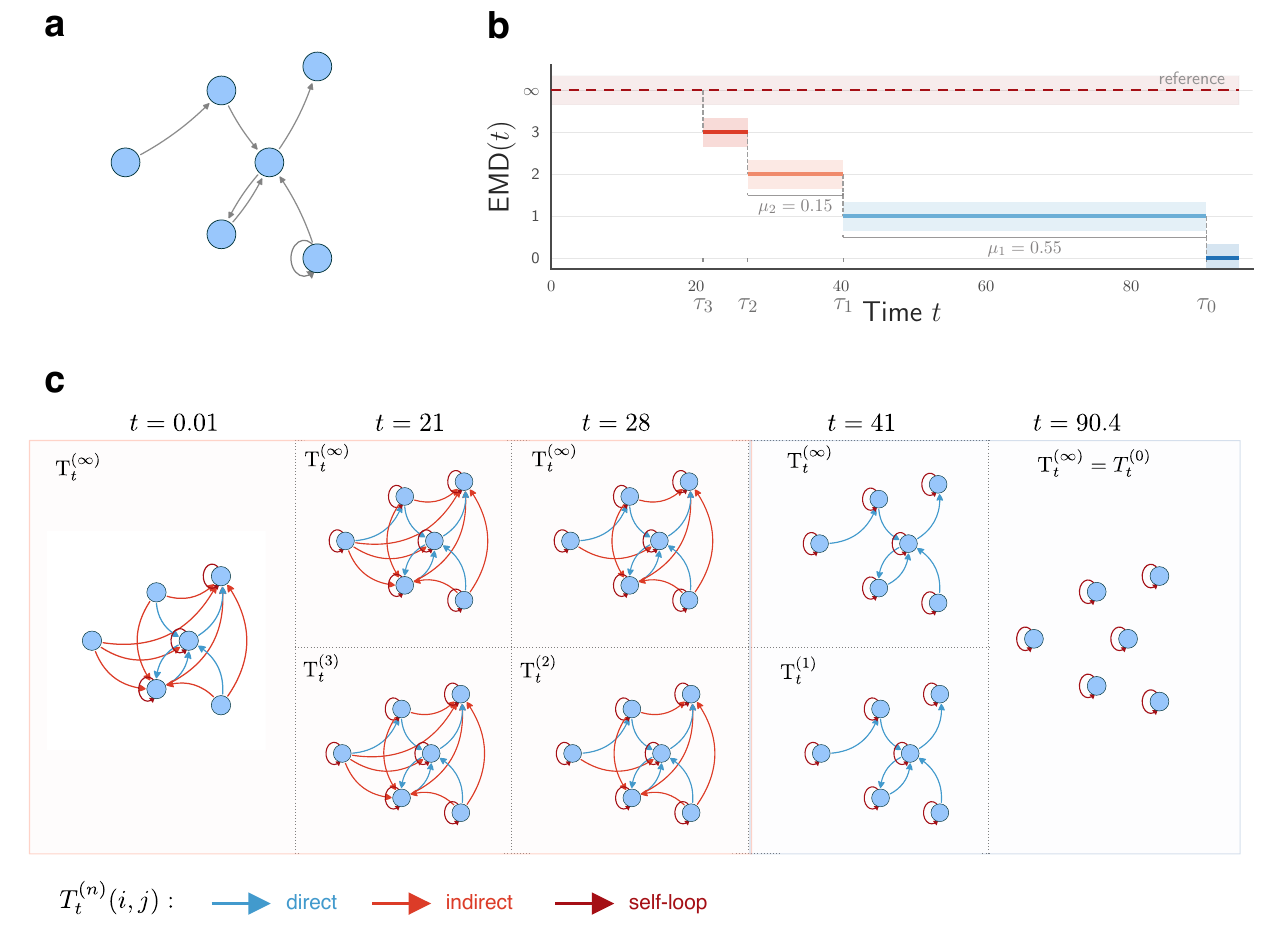}
 \caption{
\textbf{Thermodynamic cascade of hierarchical functional memory.}
\textbf{a}, Example directed network used to illustrate the formalism.
\textbf{b}, Thermalization cascade generated by progressively suppressing
long-memory contributions as the inverse temperature increases. Starting
from the infinite-history equilibrium state ($\EMD=\infty$, dark red), the
effective memory depth $\EMD(t)$ decreases in discrete steps as successive
thermalization times $\tau_n$ are crossed (vertical dashed lines). Each
horizontal band corresponds to the interval $[\tau_n,\tau_{n-1})$ during
which the depth-$n$ propagator $T_t^{(n)}$ provides the effective
description of probability flow. Colour progressively shifts from
long-memory (red) to short-memory (blue) dynamics. The hierarchical memory
spectrum $\mu_l=(\tau_{l-1}-\tau_l)/\tau_0$ records how total
thermalization time is distributed across depths, while the non-Markovian
ratio $\NMR=\tau_1/\tau_0$ measures the persistence of long-memory
relative to the freezing time $\tau_0$.
\textbf{c}, Memory flow networks generated by the full propagator
$T_t^{(\infty)}$ (top row) and by the effective truncated propagator at the
dominant active depth (bottom row) at successive macroscopic times along
the cascade. Blue edges represent direct Markovian flow, whereas red edges
represent indirect multi-hop contributions from long-memory. As
thermalization proceeds, deep indirect flows are progressively compressed
into shallower effective transitions, producing a systematic coarsening of
the flow structure. At late times the dynamics collapse onto persistent local
self-loops corresponding to frozen local-memory regime.
}
 \label{fig:cascade}
\end{figure*}

To quantify these effects, we introduce a thermodynamic framework in which a
macroscopic inverse temperature controls the depth over which past interactions
contribute to present dynamics, driving a cascade of memory-depth transitions
whose infinite-memory limit is identified with an equilibrium flow associated with 
a Kubo--Martin--Schwinger (KMS) state~\cite{Moutuou2025a,Moutuou2025b}, the equilibrium
reference state of the theory. Three observables derived from this cascade define a
signature of functional memory organisation: the effective memory depth, which tracks the
instantaneous causal horizon; the hierarchical memory spectrum, which encodes
how memory loss is distributed across depth scales; and the memory compression
rate, which quantifies how rapidly deep causal history compresses into
shallow effective dynamics. Together, these observables
resolve distinctions invisible to graph-theoretic analysis alone.

Comparing each network against null models that selectively perturb
weighted transport geometry, mesoscale structure, and edge directionality
reveals that these ingredients contribute distinct and non-equivalent
roles: weight geometry systematically governs memory depth, mesoscale
wiring shapes memory organisation across scales, and directionality
modulates dynamical response without universally determining memory depth.
Despite their structural diversity, the networks collapse onto a small
number of recurrent organisations of functional memory that remain
invisible to topological analysis alone.

\section{Results}
\subsection{Functional memory organisation as hierarchical resistance to compression}

We characterise each network through its memory architecture: the
hierarchical organisation of how causal history is retained, distributed,
and progressively compressed as the network evolves. Structural observables 
characterise connectivity and path organisation, but
they do not resolve how contributions from different depths are weighted,
suppressed, or redistributed across scales. Functional memory organisation is 
encoded precisely in these quantities.

To make memory architecture measurable, we model directed networks as
multigraphs whose asymmetric edges encode causal interactions.  Edge
weights $w_{j\to i} = A_{ij} > 0$ are either continuous positive reals,
representing interaction strengths, or positive integers, representing
multiplicities of parallel directed interactions; both cases are handled
uniformly by the framework below.  We introduce a family of
\emph{$n$-memory propagators} $\Tn$: transition matrices of a
non-Markovian random walk that, at macroscopic time $t$, samples
directed walks of length at most $n$, weighted by the energetic cost
$e^{-\beta(t)|\gamma|}$ of sustaining causal correlations across depth
$|\gamma|$. These propagators admit the compact representation
\begin{equation}
  \Rn = \sum_{k=0}^{n} e^{-\beta(t)k} A^k,
  \label{eq:resolvent_finite}
\end{equation}
with $\Tn$ obtained by column normalization, where $A$ is the
weighted adjacency matrix with entry $A_{ij} = w_{j\to i}$ for each
directed edge $j\to i$ (see Methods).

As the memory depth increases, the propagators converge toward an
infinite-history equilibrium propagator $\Tinf$ defined by
\[
  \Rinf := \lim_{n\to\infty}\Rn,
\]
representing the fully integrated dynamics of the network for $\beta(t)$
exceeding the critical inverse temperature $\beta_c = \log\rho(A)$
(Methods). This equilibrium flow is associated with KMS
states~\cite{Moutuou2025a,Moutuou2025b} and provides the thermodynamic
reference against which all finite-memory dynamics are measured.

As $\beta(t)$ increases, long-memory contributions are progressively
suppressed and the dynamics undergo a \emph{thermalization cascade}: a
sequence of memory-depth transitions in which paths exceeding successive
lengths cease to contribute appreciably to the effective flow, 
analogous to aging phenomena in glassy
systems~\cite{Mezard1986,Vincent2007}. We define the $n$-thermalization time
\begin{equation}
  \tau_n = \inf\!\left\{
    t :
    \frac{\|\Rinf - \Rn\|_F}{\|\Rinf\|_F}
    \leq \varepsilon
  \right\}
  \label{eq:tau}
\end{equation}
as the first time at which depth-$n$ memory becomes indistinguishable from
equilibrium. The ordered sequence $\tau_0 > \tau_1 > \tau_2 > \cdots$
defines the cascade, with $\tau_0$ marking the freezing time at which even
depth-1 memory has collapsed (Fig.~\ref{fig:cascade}).

Three observables derived from the cascade jointly constitute a network's
functional memory signature that does not depend on the specific choice of 
the $\b$--protocol (see Supplementary S4). The effective memory depth
\[
  \EMD(t) = \min\!\left\{
    n :
    \frac{\|\Rinf-\Rn\|_F}{\|\Rinf\|_F} \leq \varepsilon
  \right\}
\]
tracks the instantaneous effective integration depth. The hierarchical memory spectrum
\begin{equation}
  \mu(l) = \frac{\tau_{l-1}-\tau_l}{\tau_0}
  \label{eq:hms}
\end{equation}
records how thermalization time is distributed across memory depths; spectra
concentrated at large $l$ signal deep memory organisation, whereas broad
spectra reflect diffuse multiscale integration. The memory compression rate
$\gamma_n = (\tau_n - \tau_{n+1})^{-1}$ and the non-Markovian ratio
$\NMR = \tau_1/\tau_0$ condense the cascade into scalar summaries of
compression dynamics and overall memory persistence respectively.

Together, these observables define the organisation of functional memory
across scales; functional complexity, in this framework, corresponds to the
depth, distribution, and persistence of this organisation under progressive
memory compression. In this view, deep memory corresponds to resistance against the progressive
compression of long-range causal history into shallow effective dynamics.

Three structural null models serve as counterfactuals that selectively 
suppress distinct memory-generating mechanisms. The \emph{Directed
Configuration Model} (DCM) randomizes wiring while preserving the degree sequence, 
isolating mesoscale organisation. The \emph{Geometry
Fluctuation Model} (GFM) randomizes edge weights at fixed binary topology,
isolating weight geometry (see S5 in SI). The \emph{Polarity Noise Model} (PNM)
continuously interpolates edge directions toward maximum entropy, isolating
directionality (see Methods). Comparing each real network against these null ensembles
decomposes its functional memory structure into three mechanistically distinct
contributions.

\begin{figure*}[ht!]
  \centering
  \includegraphics[width=\textwidth]{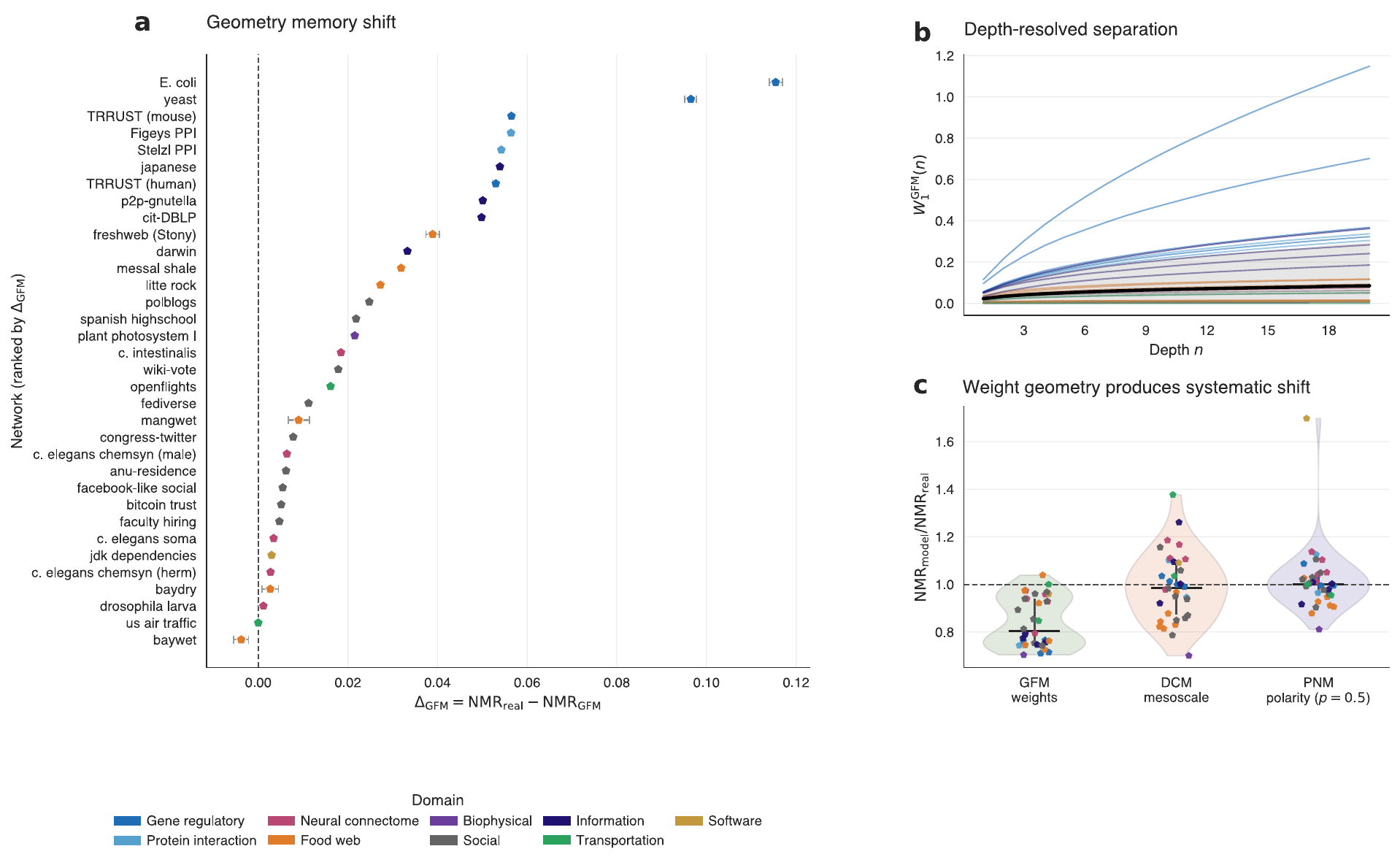}
  \caption{\textbf{A universal empirical regularity of functional memory
organisation.}
\textbf{a}, Ranked geometry-memory shift across all thirty-four empirical
networks. For each network we plot
$\Delta_{\mathrm{GFM}} = \mathrm{NMR}_{\mathrm{real}} -
\mathrm{NMR}_{\mathrm{GFM}}$, where $\mathrm{NMR}_{\mathrm{GFM}}$ is the
ensemble mean over Geometry Fluctuation Model realisations preserving the
binary topology while randomising edge weights. The dashed line marks zero.
Thirty-two networks lie strictly above zero; US air traffic lies exactly at
zero, as weight randomisation is dynamically inert for this hub-dominated
network; and one wetland food web (baywet) falls marginally below zero
($\Delta_{\mathrm{GFM}} = -0.004$, signal-to-noise ratio $0.44$), a
near-inert case with model-estimated carbon flows discussed in the main
text. Across every biological, ecological, social, and technological domain
studied, real weight geometry sustains deeper functional memory organisation
than randomised weights on the same topology. \textbf{b}, Depth-resolved
confirmation. The Wasserstein distance $W_1^{\mathrm{GFM}}(n)$ between
real and GFM hierarchical memory spectra is positive across memory depths
for all thirty-four networks, including those for which the NMR shift is
near-zero, indicating that real weight geometry consistently redistributes
memory mass toward deeper scales in the hierarchical memory spectrum.
\textbf{c}, Model specificity. Ratios $\mathrm{NMR}_{\mathrm{model}} /
\mathrm{NMR}_{\mathrm{real}}$ for GFM, DCM, and PNM show that the effect
is predominantly one-sided for weight geometry --- 32 of 34 networks have
GFM ratios below one --- whereas mesoscale and polarity perturbations
produce both increases and decreases depending on the network. Colours
denote domains.}
  \label{fig:universal}
\end{figure*}

\subsection{Real systems resist memory compression}

Comparing each network against its GFM null ensemble reveals the main empirical regularity 
identified in this work (Fig.~\ref{fig:universal}).

\medskip
\noindent
\textbf{Weight geometry consistently deepens memory.}
In thirty-two of the thirty-four networks analysed, real interaction
strengths organise functional memory at greater hierarchical depth than
randomised weight assignment on the same binary topology, with
$\NMR_{\mathrm{real}} > \NMR_{\mathrm{GFM}}$ and confidence intervals
that entirely exclude zero. US air traffic is the one system for which
weight randomisation is dynamically inert ($\NMR_{\mathrm{real}} =
\NMR_{\mathrm{GFM}}$ exactly): its thermalisation cascade is determined entirely
by hub topology, and interaction frequencies play no role in organising
memory depth. The Chesapeake Bay wet-season food web is the single near-inert case, with 
$\Delta_{\mathrm{GFM}} = -0.004$ and a signal-to-noise ratio of $0.44$; its 
carbon flows are model-estimated seasonal averages~\cite{Baird1989} rather 
than directly measured interaction strengths, and the effect sits at the 
boundary of measurement resolution for estimated continuous weights. Replacing those
model-estimated flows with unit weights recovers a strongly positive GFM
deviation ($\Delta_{\mathrm{GFM}} = +0.022$, SNR $= 11.9$;
Supplementary Section 10), confirming that the
Chesapeake Bay food web topology encodes the expected geometric memory
signal and that the near-inert result in the weighted network reflects
the resolution of the carbon flow estimates rather than a property of
the ecosystem itself. At the
depth-resolved level, the Wasserstein distance~\cite{Villani2009} between
the real hierarchical memory spectrum and the GFM ensemble is positive for
all thirty-four networks and generally increases monotonically with
reference depth, including in the two near-inert cases where the NMR
shift is negligible.

This effect holds across gene regulatory networks, protein interactomes,
neural connectomes, food webs, social platforms, citation graphs, word
adjacency networks, software dependency networks, biophysical energy
transfer, and transportation infrastructure, despite major differences in
topology, density, and scale. The comparison against the GFM ensemble
acts as a controlled intervention in which binary topology, degree
sequence, and directionality are preserved while the weighted geometry of
directed transport is randomised. The reduction of memory depth observed
in thirty-two of the thirty-four systems under this intervention shows
that real weight geometry acts as a primary structural organiser of
functional memory.

The mechanism follows directly from the structure of the resolvent
$\Rn$ (Eq.~\ref{eq:resolvent_finite}), which aggregates contributions
from directed walks across all depths up to $n$, each weighted by both
its depth and the product of its edge weights. In real systems, strong
interactions are correlated along coherent directed pathways, generating
multiplicative amplification of deep-walk contributions that survives
thermal suppression. Random weight assignment destroys these correlations,
reducing coherent amplification of deep contributions and shifting the
dynamics toward shallower effective memory. The GFM null suppresses
precisely this correlation structure while preserving all topological
constraints, showing that real systems resist the compression of deep
causal history into shallow effective dynamics. Structural analyses that
neglect edge weights therefore cannot fully recover the hierarchy of
memory organisation encoded in real interaction patterns.

This signal is absent in Erdős-Rényi and Watts-Strogatz directed graphs
with independently distributed weights, where $\Delta_{\mathrm{GFM}} < 0$
in all sixty-eight instances across both families
(Supplementary Information~\ref{sec:random_controls}), confirming that a strong positive
$\Delta_{\mathrm{GFM}}$ requires coherent alignment of interaction
strengths along directed paths rather than reflecting a topological or
computational artefact.

Beyond quantifying the effect, the GFM comparison therefore provides an operational
criterion for whether a network's edge weights encode genuine functional
interaction structure. Every system in which weight randomisation leaves
memory depth unchanged or reduces it --- including the two near-inert
cases --- exhibits a specific structural explanation: either the cascade
is determined entirely by topology (US air traffic), or the weights are
derived from a model-estimation procedure whose resolution is insufficient
to resolve the functional signal (the wet-season food web). In contrast,
every system with well-defined, directly observed interaction strengths
satisfies $\NMR_{\mathrm{real}} > \NMR_{\mathrm{GFM}}$. The GFM test
thereby identifies cases where weight construction may not faithfully
represent the functionally relevant interaction geometry, independently
of any domain-specific knowledge.

\medskip
\noindent
\textbf{Directionality modulates memory dynamics.}
The Polarity Noise Model reveals a complementary and qualitatively distinct
phenomenon. Progressive randomisation of edge directions alters the temporal
response of the cascade but does not induce a universal shift in memory
depth. The polarity slope $s = \mathrm{d}\NMR/\mathrm{d}p$ evaluated over
$p \in [0, 0.5]$ changes sign across networks: twenty systems exhibit
positive slopes, indicating that directionality suppresses deep memory
organisation, whereas fourteen exhibit negative slopes, indicating that it
sustains deep memory organisation. Unlike the geometric effect, there is no
universal directionality law.

At the spectral level, the Wasserstein response to directional perturbation
satisfies
\[
W_1^{\mathrm{PNM}} < W_1^{\mathrm{GFM}}
\]
in 29 of 34 networks, confirming that directionality usually alters the
cascade less strongly than weighted transport geometry alters the
distribution of memory across depth scales. The five exceptions are the
three neural connectomes with the strongest polarity sensitivity (the two
\textit{C.~elegans} chemical synapse networks and the \textit{Drosophila}
larval connectome) and the two directionality-dominated networks (JDK
software dependencies and Bitcoin trust), all cases where edge orientation
organises the cascade more strongly than interaction-strength geometry.
Weight geometry therefore governs how deeply functional memory is organised,
whereas directionality governs how that organisation responds to structural
perturbation. This separation between functional memory organisation and
dynamical response motivates treating structural driver and dynamical regime
as independent descriptive axes. Full $\NMR(p)$ trajectories for all
thirty-four networks are shown in Supplementary Fig. S3.

\subsection{Low-dimensional dynamical organisation of functional memory}

\begin{figure*}[ht!]
  \centering
  \includegraphics[width=\textwidth]{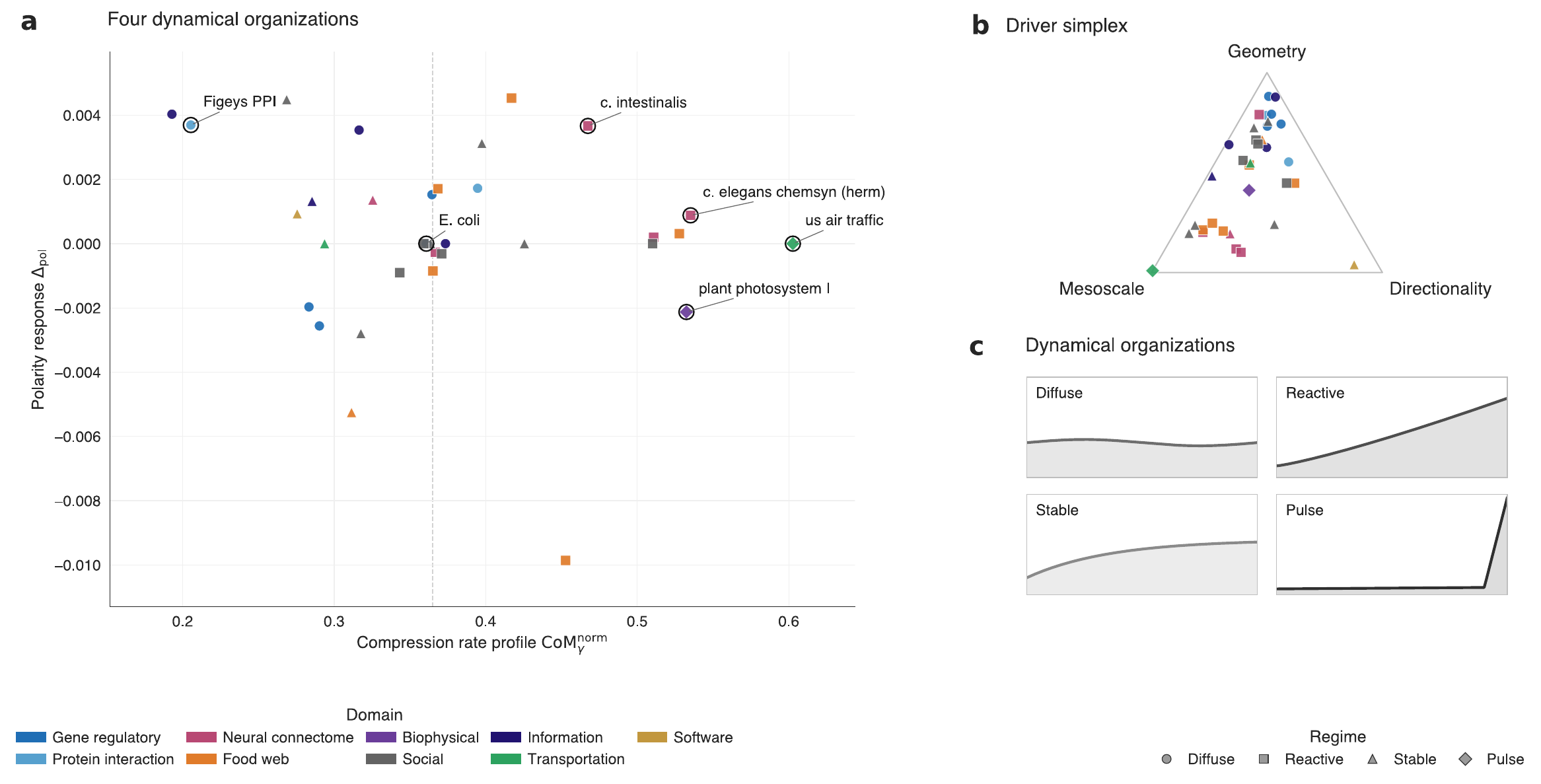}
 \caption{
\textbf{Low-dimensional organisations of functional memory across complex
systems.} Thirty-four empirical networks spanning six domains converge onto
four functional memory organisations despite major differences in topology,
scale, and biological or physical function.
\textbf{a}, Functional memory phase space. Each point represents one
network, positioned by the normalised centre of mass of its compression
rate profile, $\mathrm{CoM}_{\gamma}^{\mathrm{norm}} =
  \bigl(\sum_n n\,\gamma_n\bigr)\big/\bigl(n_{\max}\sum_n\gamma_n\bigr)$,
and by its polarity response $\Delta_{\mathrm{pol}} = \NMR(0.5) - \NMR(0)$.
The horizontal coordinate indicates the depth at which memory compression
is concentrated; the vertical coordinate measures whether directional
perturbation deepens or suppresses long-range memory organisation. Marker
shape indicates the assigned dynamical organisation; colour indicates the
domain.
\textbf{b}, Structural driver simplex. Barycentric coordinates obtained
from the normalised Wasserstein responses to the three null ensembles
quantify the relative contributions of mesoscale organisation (DCM),
weight geometry (GFM), and directionality (PNM) to the functional memory
hierarchy. Most networks cluster toward the geometry-driven sector,
consistent with the universal role of weight geometry established in
Fig.~\ref{fig:universal}.
\textbf{c}, Canonical hierarchical memory organisations. Schematic
hierarchical memory spectra and compression profiles for the four dynamical
organisations identified by the probabilistic classification described in
Methods. Diffuse systems exhibit broad, high-entropy memory distributions
spread across depth scales. Reactive systems sustain deep memory with
concentrated compression dynamics and moderate sensitivity to directional
perturbation. Stable systems sustain deep memory with more diffuse
compression profiles and weak directional sensitivity. Pulse systems
exhibit compression dominated by a single effective mode, through which
most causal history collapses at one characteristic scale. Together, these
organisations define a low-dimensional landscape of functional memory
organisation across all domains studied.
}
\label{fig:regimes}
\end{figure*}

Beyond the universal empirical regularity, the three null models jointly reveal that
functional complexity decomposes along two independent descriptive axes: the structural
ingredient responsible for organising memory, and the dynamical form taken
by that organisation (Fig.~\ref{fig:regimes}). The decomposition is particularly striking because the
networks span multiple domains, sizes, densities, and structural classes
with distinct topological constraints, making the low-dimensional organisation
itself a central result. 

For each network, normalising the three Wasserstein distances to unit sum
defines barycentric coordinates $(\phi^{\mathrm{DCM}},
 \phi^{\mathrm{GFM}},
 \phi^{\mathrm{PNM}})$ on a ternary diagram
(Fig.~\ref{fig:regimes}b) whose vertices represent memory hierarchies
organised predominantly by mesoscale wiring, weight geometry, or
directionality (see Methods).

\medskip
\noindent
\textbf{Structural driver axis.}
Weight geometry is the dominant Wasserstein coordinate in twenty-one of the
thirty-four networks, spanning biological, ecological, social, informational,
biophysical, and technological domains. Ten networks are mesoscale-dominated,
indicating that their memory hierarchy is shaped primarily by higher-order
wiring organisation beyond the degree sequence. One network, the Darwin
correspondence graph, has near-balanced DCM and GFM contributions and is
classified as jointly geometry- and mesoscale-dominated. Two systems, the
JDK software dependency graph and the Bitcoin trust network, are
directionality-dominated, indicating that polarity perturbation produces the
largest reorganisation of their memory hierarchy. The driver axis identifies
which structural ingredient must be preserved or perturbed to alter a
network's functional memory organisation, and its assignment is not
recoverable from topology alone.

\medskip
\noindent
\textbf{Dynamical organisation axis.}
Independent from the structural driver axis, the shape of the hierarchical
memory spectrum together with the compression dynamics separates networks
into four recurrent dynamical organisations
(Fig.~\ref{fig:regimes}).

One class of systems exhibits diffuse multiscale organisation, with memory
distributed broadly across depth scales and no sharply preferred integration
horizon. This group includes strongly geometry-driven systems such as
\textit{E.~coli}, both human protein interactomes, both TRRUST regulatory
networks, and the yeast transcriptome~\cite{ShenOrr2002,Ewing2007,Stelzl2005,Han2015}.
These networks sustain high absolute non-Markovian ratios primarily through
weight geometry, while directionality contributes variably and without a
common signature, suggesting that these systems sustain causal integration
across many depth scales simultaneously.

A second organisation concentrates memory at large hierarchical depth with
concentrated compression dynamics and moderate directional sensitivity.
This class contains five food webs together with four neural connectomes:
the somatic \textit{C.~elegans} connectome~\cite{Varshney2011,Cook2019,Moutuou2025a},
the hermaphrodite chemical synapse network~\cite{Cook2019},
the \textit{C.~intestinalis} central nervous
system~\cite{Ryan2016}, and the \textit{Drosophila} larval
connectome~\cite{Winding2023}. Despite major biological and structural
differences, trophic networks and neural circuits repeatedly converge onto
a common dynamical organisation characterised by sustained deep memory and
concentrated compression, with directional structure modulating the
persistence of long-range causal integration.

A third organisation sustains deep memory with comparatively diffuse
compression profiles and weak polarity reactivity. This class includes
systems whose long-range integration is stabilised either by mesoscale
architecture or by weighted transport geometry. Examples include global
aviation, faculty hiring, the word adjacency network of text from Darwin's
\textit{The Origin of Species}~\cite{Milo2004}, the male
\textit{C.~elegans} chemical synapse network~\cite{Cook2019}, and
the Messel Shale food web~\cite{Dunne2014}. Although these systems differ
substantially in domain and structural driver, they share a common
dynamical signature: deep memory persists while remaining relatively
insensitive to directional perturbation.

The final organisation exhibits extreme memory compression concentrated at
a single characteristic depth. Its compression profile is dominated by a
single pronounced peak, indicating that most causal history is compressed
through one effective memory scale. US air traffic and plant
Photosystem~I realise this organisation through physically distinct
mechanisms: one associated with extreme concentration of transport flows
around major hubs, the other with highly constrained energy-transfer
pathways. Despite belonging to unrelated physical domains, both systems
display memory organisation dominated by a single compression scale,
making all other contributions comparatively secondary.

Each network is therefore characterised by two independent descriptors: a
structural driver identifying which ingredient primarily organises its
memory hierarchy, and a dynamical organisation describing how memory is
distributed across depth scales. Together, these axes provide a
low-dimensional description of functional memory organisation across all
domains studied.

\medskip
\noindent
\textbf{Signed role of directionality.}
The ternary PNM coordinate measures the magnitude of directionality's
contribution to the memory hierarchy (Fig.~\ref{fig:regimes}b). To determine its sign, we define the
polarity response $\Delta_{\mathrm{pol}} = \NMR(p{=}0.5) - \NMR(p{=}0)$,
where $\NMR(p)$ denotes the non-Markovian ratio after polarity perturbation
at level $p$. Negative values indicate that directionality sustains
deep-memory organisation, positive values that it suppresses deep memory organisation,
and values near zero that directionality is dynamically neutral. This signed
response is independent of the magnitude measured by the ternary coordinate. 
Two networks can therefore exhibit comparable directionality contributions
while responding oppositely to polarity perturbation, further separating
the organisation of memory from its dynamical response to structural
perturbation.

\subsection{Cross-domain convergence and hidden functional divergence}

The two-axis decomposition reveals both hidden functional divergence
within shared structural classes and unexpected convergence across
unrelated domains. Networks that conventional structural analysis would
group together can occupy opposite regions of functional memory organisation space,
whereas systems with no common topology, scale, or substrate can converge
onto the same dynamical organisation. The convergence is not imposed by
the classification itself, but emerges directly from the cascade
observables derived from the thermodynamic framework.

\medskip
\noindent\textbf{Functional divergence within a shared biological class.}
The clearest illustration is provided by the \textit{E.~coli}
transcriptional regulatory network~\cite{ShenOrr2002} and the Figeys human protein
interactome~\cite{Ewing2007}. Both networks are
geometry-dominated and belong to the Diffuse organisation, indicating that
their hierarchical memory spectra are broadly distributed across depth
scales and organised primarily by weighted transport geometry. Yet they diverge in both mesoscale 
deviation and polarity response.

The \textit{E.~coli} network achieves a high absolute memory level
($\NMR \approx 0.398$) and lies marginally above the DCM null ensemble
($\Delta_{\mathrm{DCM}} \approx +0.004$), indicating that its mesoscale
organisation contributes slightly beyond what is encoded by the degree sequence alone.
Its large GFM Wasserstein distance ($W_1^{\mathrm{GFM}} \approx 1.15$, the
largest in the dataset) shows that weight geometry is the primary
organiser of its hierarchical memory structure. The polarity slope is
weakly negative ($s\approx -0.003$), indicating that directionality slightly sustains 
deep memory organisation.

The Figeys interactome exhibits the opposite trend. Although it remains geometry-dominated 
and Diffuse, it falls below the DCM null ensemble ($\Delta_{\mathrm{DCM}} \approx -0.022$),
indicating that its higher-order wiring organisation suppresses deep memory organisation 
relative to degree-sequence expectation. Its polarity slope is
strongly positive ($s\approx +0.047$), showing that directional perturbation increases memory depth rather than
suppressing it. Both the mesoscale deviation and the polarity response
therefore reverse sign relative to \textit{E.~coli}, revealing a divergence that conventional 
structural observables do not resolve.

\medskip

\noindent\textbf{Convergence of trophic systems.}
The food webs provide a particularly clear example of convergence at the
level of functional organisation despite diversity in structural drivers.
The six systems span estuarine, freshwater, and palaeontological ecosystems
separated by large differences in species count and geological timescale.
Five of the six exhibit positive GFM deviations, indicating that weighted
transport geometry deepens functional memory in those trophic systems. The
exception is the Chesapeake Bay wet-season food web~\cite{Baird1989},
whose near-zero effect ($\Delta_{\mathrm{GFM}} = -0.004$, signal-to-noise
ratio $0.44$) reflects the measurement resolution of its model-estimated
carbon flows rather than a genuine reversal of the geometric effect; replacing those 
flows with unit weights recovers a strongly positive
deviation ($\Delta_{\mathrm{GFM}} = +0.022$, SNR $= 11.9$), confirming
that the Chesapeake Bay trophic topology itself encodes the expected
geometric memory signal.
All six also exceed their DCM null ensemble, indicating that higher-order
wiring organisation beyond the degree sequence contributes an additional
positive increment to memory depth across every trophic system analysed.

The driver assignments split evenly: the three estuarine and wetland
systems (Chesapeake Bay dry-season, Chesapeake Bay wet-season, and the
mangrove wetland) are mesoscale-dominated, whereas the two freshwater
systems (Little Rock Lake and Stony Creek) and the palaeontological
system (Messel Shale) are geometry-dominated. Both interaction-strength
geometry and higher-order trophic wiring therefore contribute to memory
organisation, with the dominant ingredient tracking broad ecosystem type.

Despite this driver divergence, all six food webs occupy the same region
of dynamical organisation space. Five belong to the Reactive organisation
--- the three mesoscale-dominated wetland systems together with Little Rock
Lake and Stony Creek --- and one (Messel Shale) to the Stable organisation.
The three mesoscale-dominated food webs fall exclusively in Reactive,
while the three geometry-dominated systems split two-to-one between the
two deep-memory classes, with Stony Creek classified as borderline
Reactive ($p_{\max} = 0.50$).

A further convergence is observed in polarity response. Five of the six
food webs exhibit negative polarity slopes, indicating that edge
directionality sustains rather than suppresses deep memory organisation
in those trophic systems. The exception is the Chesapeake Bay wet-season
food web, whose slope of $+0.004$ is effectively zero, indicating polarity
neutrality consistent with its near-inert GFM response. This predominantly
negative directional signature is independent of structural driver and
dynamical organisation class, and distinguishes trophic systems as a group
from neural connectomes, where polarity responses are uniformly positive,
and gene regulatory networks, where they are predominantly positive.
Together, the shared deep-memory dynamical organisation and the dominant
negative polarity response suggest that trophic flow direction encodes a
fundamental principle of causal memory retention in ecological systems,
regardless of whether that memory is organised primarily by weighted
transport geometry or by higher-order wiring
structure~\cite{Williams2000,Dunne2002}.

\noindent\textbf{Neural connectomes: an exception to geometry dominance.}
Four of the five neural connectomes in the dataset are mesoscale-dominated,
indicating that higher-order circuit architecture contributes more strongly
than weighted transport geometry to the organisation of memory across depth
scales. This pattern holds for the somatic \textit{C.~elegans}
connectome, the hermaphrodite and male chemical synapse connectomes, and
the \textit{Drosophila} larval connectome, in all of which the DCM
Wasserstein deviation exceeds the GFM deviation, confirming that motifs,
feedback circuits, and hierarchical routing architecture organise memory
more strongly than interaction strengths alone.

\textit{C.~intestinalis} is the exception. Its ternary position places it
firmly in the geometry-dominated sector ($\phi^{\mathrm{GFM}} \approx 0.79$,
ranking sixth among all geometry-dominated networks in the dataset), with
the GFM Wasserstein deviation exceeding the DCM deviation by nearly
six-fold. Weight geometry, rather than higher-order wiring architecture,
is therefore the primary organiser of functional memory in this protochordate
connectome.

The dynamical organisations split across both driver groups. The somatic
\textit{C.~elegans} connectome, the hermaphrodite chemical synapse network,
the \textit{Drosophila} larval connectome, and \textit{C.~intestinalis}
all belong to the Reactive organisation, whereas the male
\textit{C.~elegans} chemical synapse network belongs to Stable.
Reactive membership is therefore conserved across the mesoscale-dominated
and geometry-dominated neural connectomes alike, indicating that deep memory
with concentrated compression dynamics is a shared dynamical signature of
neural circuits, independent of which structural ingredient primarily
organises that memory.

A further signature is shared by all five connectomes: every polarity slope
is positive, indicating that directionality universally suppresses rather
than sustains deep memory organisation in neural circuits. This is in
contrast to food webs, where directionality predominantly sustains deep
memory organisation, with only one system showing polarity neutrality.
Neural connectomes therefore remain a structurally coherent group at the
level of dynamical regime and polarity response, even as they diverge in
structural driver.

\begin{figure*}[th!]
    \includegraphics[width=\textwidth]{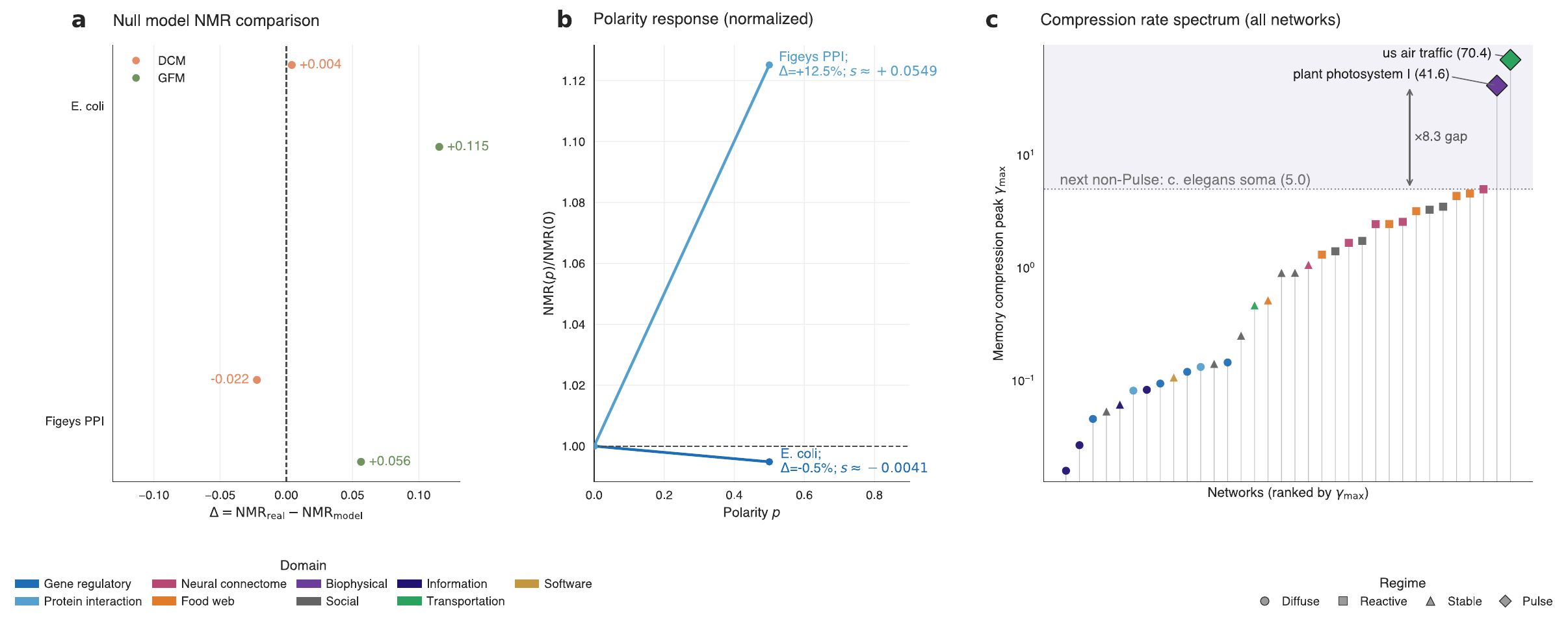}
    \caption{\textbf{Hidden functional divergence and cross-domain convergence
revealed by thermodynamic decomposition.}
\textbf{a}, Null-model comparison for the \textit{E.~coli}
transcriptional regulatory network and the Figeys human protein
interactome, two systems sharing the same geometry-dominated driver and
Diffuse memory organisation. The horizontal axis shows
$\Delta=\NMR_{\mathrm{real}}-\NMR_{\mathrm{model}}$ for the DCM (orange)
and GFM (green) null ensembles; the real-network NMR is anchored at the
dashed reference ($x=0$). Error bars indicate 95\% bootstrap confidence
intervals. Despite their identical classification, the two networks diverge
systematically: \textit{E.~coli} lies marginally above the DCM ensemble
($\Delta_{\mathrm{DCM}}\approx +0.004$) and well above the GFM ensemble
($\Delta_{\mathrm{GFM}}\approx +0.115$), whereas Figeys lies below the
DCM ensemble ($\Delta_{\mathrm{DCM}}\approx -0.022$) while remaining above
the GFM ensemble ($\Delta_{\mathrm{GFM}}\approx +0.056$).
\textbf{b}, Normalised polarity response $\mathrm{NMR}(p)/\mathrm{NMR}(0)$
as a function of polarity perturbation $p\in[0,0.5]$ for the same two
networks. \textit{E.~coli} is nearly invariant to directional perturbation
($s\approx -0.003$), whereas Figeys exhibits a strong positive response
($s\approx +0.047$), increasing by approximately $13\%$ at $p=0.5$.
The opposite polarity responses are orthogonal to the shared Diffuse
organisation and are not captured by conventional structural descriptors.
\textbf{c}, Maximum compression rate $\gamma_{\max}$ for all thirty-four
networks, ranked in ascending order (log scale). The two Pulse systems,
US air traffic ($\gamma_{\max}\approx 70.4$) and plant Photosystem~I
($\gamma_{\max}\approx 41.6$), are sharply separated from the remainder of
the dataset despite sharing no common topology, scale, or physical
substrate. One system transports passengers through a global mobility
network, whereas the other transports excitonic energy through a
photosynthetic complex. Their convergence onto the same compression
organisation illustrates the strongest example of cross-domain functional
similarity observed in the dataset. Marker colour denotes domain and marker
shape denotes dynamical organisation.}
 \label{fig:convergence-vs-divergence}
\end{figure*}

\medskip 
\noindent\textbf{Extreme convergence across unrelated physical substrates.}
The strongest convergence in the dataset occurs between plant
Photosystem~I~\cite{Mohseni2008}
($\NMR \approx 0.073$, $\Delta_{\mathrm{DCM}} \approx +0.022$,
$\Delta_{\mathrm{GFM}} \approx +0.022$,
$\gamma_{\max} \approx 41.6$)
and US air traffic~\cite{Barrat2004}
($\NMR \approx 0.045$, $\Delta_{\mathrm{DCM}} \approx -0.017$,
$\Delta_{\mathrm{GFM}} = 0$, $\gamma_{\max} \approx 70.4$). Despite one
system transporting excitonic energy and the other transporting passengers,
both collapse onto the Pulse organisation.

In both cases, compression dynamics are dominated by a single characteristic
compression scale. US air traffic exhibits the largest compression peak in
the dataset ($\gamma_{\max}\approx 70.4$), while Photosystem~I
($\gamma_{\max}\approx 41.6$) is the only other system in the same range.
Both are sharply separated from all non-Pulse networks, whose largest
compression peak is approximately $5.0$
(Fig.~\ref{fig:convergence-vs-divergence}). The two systems differ in
their structural drivers and null-model responses: US air traffic is
mesoscale-dominated and inert under GFM, whereas Photosystem~I is
geometry-dominated and remains separated from both DCM and GFM ensembles.
Their convergence is therefore not structural, but dynamical: both organise
functional memory through an extreme concentration of compression at one
dominant scale.

Taken together, these examples show that systems organise according to
thermodynamic memory structure rather than structural similarity alone.
Networks with similar biological class and shared memory organisation can
diverge in mesoscale and polarity response, whereas systems with unrelated
substrates can converge onto the same organisation of functional memory.
The thermodynamic decomposition thereby reveals organisational principles
that remain invisible to conventional structural network analysis.

\section{Discussion}
This work identifies functional memory organisation as a measurable and
cross-domain dimension of complexity in directed networks. In thirty-two of
the thirty-four empirical systems analysed, real interaction strengths
organise memory at greater hierarchical depth than randomised weighted
transport geometry on the same topology, with confidence intervals that
entirely exclude zero. A thirty-third system, US air traffic, satisfies
the relation with equality ($\Delta_{\mathrm{GFM}} = 0$ exactly): its
thermalisation cascade is determined entirely by hub topology, and
interaction frequencies carry no additional memory structure beyond what
binary connectivity provides. The Chesapeake Bay wet-season food web is
the one genuine exception, with $\Delta_{\mathrm{GFM}} = -0.004$; its
carbon flows are model-estimated seasonal averages~\cite{Baird1989} rather
than directly measured interaction strengths, and binarising the network
recovers a strongly positive GFM deviation ($\Delta_{\mathrm{GFM}} = +0.022$,
SNR $= 11.9$), confirming that the exception reflects the resolution of
the flow estimates rather than a property of the ecosystem. The mechanism
is the same in all positive cases: coherent weight geometry amplifies deep
contributions along paths, whereas randomisation suppresses these
correlations and shifts the dynamics toward shallower effective memory.
Functional organisation is therefore encoded not only in which interactions
exist, but also in how their strengths sustain long-range causal integration
under progressive memory compression.

These results clarify the limits of topology alone. Binary architecture
determines the set of admissible paths, but it does not determine how
strongly distant histories remain integrated through them. The Geometry
Fluctuation Model isolates the contribution of weighted transport geometry
by preserving the binary skeleton while perturbing interaction strengths.
Across all domains where weight geometry is not dynamically inert, this
perturbation systematically reduces memory depth, demonstrating that
weighted organisation contains information about functional memory that is
not recoverable from topological structure alone. Treating weights as
secondary refinements or discarding them through binarisation removes a
primary determinant of how real systems organise causal history across
scales.

Beyond its role as a comparative tool, the GFM comparison provides an
operational criterion for whether a network's edge weights encode
functionally meaningful interaction structure. A positive GFM deviation
indicates that the specific assignment of interaction strengths on a given
topology deepens causal memory beyond what random weight assignment
produces, confirming that the weights carry functional information not
recoverable from binary connectivity alone. A near-zero or negative
deviation, by contrast, signals that the weights are effectively inert for
cascade memory organisation: either the cascade is determined entirely by
topology, as in US air traffic, or the weight values are model-estimated
proxies whose resolution does not suffice to distinguish their effect from
random assignment on the same skeleton, as in the wet-season food web.
Applied prospectively, the GFM test therefore provides a principled
diagnostic for whether a proposed weight configuration faithfully represents
the functionally relevant interaction geometry --- an assessment that cannot
be made from structural properties of the network alone.

The thermodynamic formulation provides a common framework for comparing
systems with otherwise unrelated substrates and dynamics. The
infinite-memory limit defines the equilibrium reference relative to which
finite-memory dynamics are measured, and the cascade times $\tau_n$
identify the stages at which correlations beyond depth $n$ become
effectively redundant. Memory loss is therefore resolved as a hierarchy of
compression events rather than a single relaxation scale. The effective
memory depth, hierarchical memory spectrum, and compression profile
together quantify how deeply systems integrate causal history, how that
memory is distributed across depths, and how rapidly it compresses into
shallow effective dynamics.

Beyond the universal role of weight geometry, the null-model decomposition
reveals a low-dimensional organisation of functional memory. Despite
splitting evenly between geometry-dominated and mesoscale-dominated
structural drivers, all six food webs occupy deep-memory dynamical
organisations, and directionality predominantly sustains rather than
suppresses causal integration across trophic systems, a signature that
holds regardless of ecosystem type or geological timescale. Among neural
connectomes, the directional response is reversed: directionality
suppresses deep memory in all five circuits, whether the primary
structural organiser is mesoscale wiring, as in four of the five
connectomes, or weight geometry, as in \textit{C.~intestinalis}.
Photosystem~I and US air traffic, despite transporting entirely different
physical quantities, both exhibit extreme compression concentrated at a
single characteristic scale, converging on the Pulse organisation through
physically distinct mechanisms. None of these correspondences follow from
domain labels or topological similarity; all emerge from the cascade
observables of the thermodynamic framework.

The dataset is intentionally heterogeneous rather than exhaustive, and
this heterogeneity is central to the empirical result. The geometric
effect persists across gene regulation, neural circuits, ecological webs,
social communication, citation networks, software dependencies, biophysical
energy transfer, and transportation infrastructure, spanning several orders
of magnitude in size and density. The framework also generates falsifiable
predictions. In systems where directionality suppresses deep memory
($s > 0$), complete edge reversal ($p = 1$) should partially restore the
original directionality structure and recover a non-Markovian ratio close
to its unperturbed value, whereas partial scrambling ($p = 0.5$) produces
the largest deviation. Quantitatively, the recovery ratio
$(\NMR(p{=}1)-\NMR(p{=}0))/s$ is predicted to anticorrelate with the
directional asymmetry index $(\|A-A^\top\|_F)/\|A\|_F$, because highly
asymmetric systems possess less reversal symmetry to exploit. This
prediction links the dynamical organisation axis to a measurable
topological property and provides an independent test of whether the
observed organisations capture functionally meaningful distinctions rather
than statistical regularities.

Several extensions follow naturally from the framework. Temporal
networks~\cite{Holme2012} would allow interaction strengths to evolve
independently of the thermodynamic control parameter, making it possible
to study how organisations of functional memory persist or reorganise as
weight geometry changes over time. Multilayer
networks~\cite{Kivela2014,Moutuou2023} raise the question of whether
different interaction channels sustain distinct memory organisations that
compete or reinforce across layers. The framework also suggests connections
with information-theoretic approaches to complexity. The cascade times
$\tau_n$ measure how many depth scales are required before
infinite-history dynamics admit an effective finite representation. This is
structurally analogous to logical depth~\cite{Bennet1990},
Kolmogorov--Sinai entropy~\cite{Lloyd2001}, and algorithmic
complexity~\cite{Antunes2006}, where complexity reflects the cost of
recovering structure from compressed descriptions. Establishing this
connection formally could clarify why the space of functional memory
organisation appears intrinsically low-dimensional across real systems.

\section{Methods}

\subsection{Memory propagators and the thermalization cascade}
A directed network is modelled as a weighted directed graph
$G = (V, E)$ with adjacency matrix $A$, where $A_{ij} \geq 0$ encodes
the interaction strength from node $j$ to node $i$. This convention
makes columns correspond to sources and rows to targets; self-loops are
permitted. Two representations are used depending on the network.
For networks with integer-valued interaction multiplicities, $A_{ij}$
counts the number of parallel directed edges from $j$ to $i$, and $G$
is a directed multigraph. For networks with real-valued interaction
strengths (such as the food web carbon flows used here) $A_{ij}$
records a continuous non-negative weight encoding the magnitude of the
interaction. The thermodynamic framework applies identically in both
cases: the matrix power $(A^k)_{ij}$ aggregates contributions from
all directed walks of length $k$ from $j$ to $i$, weighted by
interaction multiplicities or strengths, and the resolvent construction
below is unchanged.

The framework distinguishes two timescales. A microscopic timescale
governs probability flow through the propagators $T_t^{(n)}$, whereas
a slow macroscopic timescale $t$ controls the suppression of memory
depth through an inverse-temperature parameter $\beta(t)$. At fixed
macroscopic time, $\beta(t)$ is treated as quasi-static~\cite{Callen1985}, and the
propagators describe the effective microscopic dynamics under that
thermodynamic condition.

For a memory cutoff $n$, the corresponding truncated resolvent is
\begin{equation}
  R_t^{(n)} = \sum_{k=0}^{n} e^{-\beta(t)k} A^k,
  \label{eq:resolvent_finite_methods}
\end{equation}
where $R_t^{(n)}(i,j)$ is the aggregated contribution of all directed walks from $j$
to $i$ of length at most $n$, with a walk of length $k$ penalized by
$e^{-\beta(t)k}$. The $n$-memory propagator is obtained by column
normalization:
\begin{equation}
  T_t^{(n)}(i,j)
  = \frac{R_t^{(n)}(i,j)}{\sum_{i'} R_t^{(n)}(i',j)},
  \label{eq:Tn_methods}
\end{equation}
whenever the denominator is non-zero. Thus $T_t^{(n)}$ is the effective transition 
matrix associated with a finite-memory random walk that samples paths of 
length at most $n$. 

The propagators therefore describe microscopic flow dynamics conditioned on
a slowly varying thermodynamic environment parametrized by $t$.

Let $\rho(A)$ denote the spectral radius of $A$ and set
$\beta_c = \log\rho(A)$. For $\beta(t) > \beta_c$, the Neumann series
converges to the full resolvent:
\begin{equation}
  R_t^{(\infty)} = \sum_{k=0}^{\infty} e^{-\beta(t)k} A^k
  = \bigl(I - e^{-\beta(t)} A\bigr)^{-1}.
  \label{eq:resolvent_full_methods}
\end{equation}
After normalization, this defines the infinite-history propagator
$T_t^{(\infty)}$. In the graph $C^\ast$-algebraic formulation, this
equilibrium propagator corresponds to the KMS state at inverse temperature
$\beta(t)$~\cite{Moutuou2025a,Moutuou2025b}. Note that the resolvent is the
operator object from which this equilibrium propagator is obtained;
it is not itself the KMS state.

As the macroscopic control parameter $t$ increases, the corresponding
inverse temperature $\beta(t)$ progressively suppresses contributions from
long walks. For a fixed tolerance $\varepsilon > 0$, the $n$-thermalization time is
\begin{equation}
  \tau_n = \inf\!\left\{
    t : \frac{\|R_t^{(\infty)} - R_t^{(n)}\|_F}{\|R_t^{(\infty)}\|_F}
    \leq \varepsilon
  \right\}.
  \label{eq:tau_methods}
\end{equation}
The ordered sequence $\tau_0 \ge \tau_1 \ge \tau_2 \ge \cdots$ defines the
thermalization cascade and records the macroscopic times at which
successively deeper memory cutoffs become sufficient to approximate the
infinite-history propagator. The value $\varepsilon = 10^{-5}$ was used
throughout (see S8 in SI).

The effective memory depth is the minimal cutoff needed at time $t$:
\begin{equation}
  \mathrm{EMD}(t) = \min\!\left\{
    n \geq 0 : \frac{\|R_t^{(\infty)} - R_t^{(n)}\|_F}{\|R_t^{(\infty)}\|_F}
    \leq \varepsilon
  \right\},
  \label{eq:emd_methods}
\end{equation}
a non-increasing step function of $t$, takes the value $n$ on the interval
$[\tau_n, \tau_{n-1})$.

In this formulation, the cascade quantifies how rapidly infinite-history
dynamics admit finite-depth effective descriptions under progressive memory
suppression.

\subsection{Randomized estimation of effective memory depth}
For large networks, explicitly forming the infinite-memory resolvent
$R_t^{(\infty)}$ becomes computationally expensive and numerically
ill-conditioned near the critical inverse temperature $\beta_c$. We
therefore estimate the relative Frobenius error in
equation~\eqref{eq:tau_methods} using randomized trace
estimation~\cite{Hutchinson1990,Avron2011}.

Let $z \in \mathbb{R}^N$ be a random probe vector with independent
Rademacher entries
\[
\mathbb{P}(z_i = 1)=\mathbb{P}(z_i=-1)=\frac12.
\]
For any matrix $M$,
\begin{equation}
  \|M\|_F^2
  = \mathrm{Tr}(M^\top M)
  = \mathbb{E}\!\left[\|Mz\|_2^2\right].
  \label{eq:hutchinson_methods}
\end{equation}
The squared relative Frobenius error can therefore be approximated by
Monte Carlo averaging over $m$ independent probe vectors:
\begin{equation}
  \frac{\|R_t^{(\infty)} - R_t^{(n)}\|_F^2}
       {\|R_t^{(\infty)}\|_F^2}
  \approx
  \frac{
    m^{-1}\sum_{j=1}^{m}
    \left\|
      (R_t^{(\infty)} - R_t^{(n)})z^{(j)}
    \right\|_2^2
  }{
    m^{-1}\sum_{j=1}^{m}
    \left\|
      R_t^{(\infty)} z^{(j)}
    \right\|_2^2
  }.
  \label{eq:ratio_estimator_methods}
\end{equation}
Although the cascade definitions in
equations~\eqref{eq:tau_methods}--\eqref{eq:emd_methods}
use the unsquared relative Frobenius norm, the squared estimator preserves
the ordering of thermalization times and effective memory depths while
improving numerical efficiency.

The action of the infinite-memory resolvent on a probe vector is computed
iteratively through the Neumann recursion
\[
v_{k+1}=e^{-\beta(t)}A\,v_k,
\qquad
v_0=z,
\]
while accumulating the partial sums
\[
\sum_{k=0}^{K} v_k.
\]
The truncated action $R_t^{(n)}z$ is obtained by stopping the accumulation
at depth $n$. The iteration converges whenever
\[
e^{-\beta(t)}\rho(A)<1,
\]
equivalently whenever $\beta(t)>\beta_c$.

All reported computations used $m=20$ independent probe vectors.
Increasing the probe count produced negligible changes in the resulting EMD
curves across the empirical networks. The estimator therefore provides an
efficient approximation of how rapidly infinite-history dynamics admit
finite-depth effective representations under progressive memory
suppression.

\subsection{Hierarchical memory spectra and Wasserstein comparison}

The hierarchical memory spectrum of a network is the probability
distribution
\begin{equation}
  \mu_G(l)
  =
  \frac{\tau_{l-1}-\tau_l}{\tau_0},
  \qquad l\ge1,
  \label{eq:hms_methods}
\end{equation}
normalized so that
\[
\sum_{l\ge1}\mu_G(l)=1.
\]
The spectrum records how the total thermalization time is distributed
across memory depths. Large tail mass indicates that thermalization
remains distributed across deep memory scales, whereas broad
high-entropy spectra indicate diffuse multiscale organisation across the
hierarchy.

To compare memory organisation across networks and null ensembles, we use
the 1-Wasserstein distance~\cite{Villani2009} on the ordered depth support
$\{1,\ldots,l_{\max}\}$. For probability measures $\mu$ and $\nu$,
\begin{equation}
  W_1(\mu,\nu)
  =
  \sum_{l=1}^{l_{\max}}
  \left|
    \sum_{k=1}^{l}\mu(k)
    -
    \sum_{k=1}^{l}\nu(k)
  \right|,
  \label{eq:wasserstein_methods}
\end{equation}
corresponding to the $L^1$ distance between cumulative distribution
functions on the depth axis. This metric is naturally adapted to
hierarchical memory spectra because memory depth is intrinsically
ordered: redistributing spectral mass from depth $i$ to depth $j$
incurs a cost proportional to $|i-j|$. Wasserstein distances therefore
quantify not only whether two systems differ in memory organisation, but
also the depth scales over which those differences are distributed.

For each empirical network, spectra are compared against ensembles
generated by the null models described below. When ensemble means are
reported, Wasserstein distances are computed relative to the average
null spectrum across realizations; when confidence intervals are needed,
the full distribution of null realizations is retained.

To resolve how deviations evolve across the hierarchy, we compute
depth-resolved curves
\[
W_1^{\mathrm{model}}(n),
\]
obtained by truncating the spectra at successive reference depths $n$.
Because the cascade progressively suppresses deeper contributions first,
deviations at large reference depths correspond to resistance against
compression of long-range causal history. The resulting curves therefore
provide a depth-resolved measure of resistance to memory compression
across the hierarchy.

\subsection{Thermodynamic protocol}
Recall that the framework distinguishes two timescales. Microscopic probability flow is
described by the propagators $T_t^{(n)}$, whereas a slow macroscopic
timescale governs the progressive suppression of memory depth through the
inverse-temperature parameter $\beta(t)$. At fixed macroscopic time, the
propagators evolve under quasi-static thermodynamic conditions determined by
$\beta(t)$.

The macroscopic evolution of memory suppression is governed by the protocol
\begin{equation}
  \beta(t)
  =
  \beta_c
  +
  \left[
    \log\!\left(1+\frac{t}{\kappa}\right)
  \right]^{\alpha},
  \label{eq:protocol_methods}
\end{equation}
with
\[
\kappa=\frac{N}{\bar d},
\qquad
\alpha
=
1+\frac{1}{\log(1+\bar d)},
\]
where $N$ is the number of nodes and $\bar d$ is the mean out-degree. The
protocol satisfies $\beta(0)=\beta_c$ and increases monotonically with
$t$, progressively suppressing contributions from long walks while
resolving early memory-depth transitions over extended macroscopic times.

The normalization
\[
\kappa=\frac{N}{\bar d}
\]
rescales the macroscopic clock according to network sparsity, allowing
networks with different sizes and connectivities to be compared on
comparable thermodynamic timescales. Sparse networks therefore evolve more
slowly under the protocol than dense networks of similar size.

The logarithmic dependence on $t$ is motivated by aging phenomena in
disordered trap hierarchies and glassy relaxation
processes~\cite{Bouchaud1992,Monthus1996}, where equilibration occurs
through progressively slower exploration of deeper configurational states.

For each memory depth $n$, the corresponding protocol-independent
inverse-temperature threshold is
\begin{equation}
  \beta_n
  =
  \inf\!\left\{
    \beta>\beta_c :
    \frac{
      \|R_\beta^{(\infty)}-R_\beta^{(n)}\|_F
    }{
      \|R_\beta^{(\infty)}\|_F
    }
    \le\varepsilon
  \right\},
  \label{eq:beta_thresh_methods}
\end{equation}
which depends only on the network adjacency matrix $A$ and the tolerance
$\varepsilon$. The thermalization times are then obtained through the
macroscopic clock relation
\[
\tau_n=\beta^{-1}(\beta_n).
\]
The sequence $\{\beta_n\}$ therefore characterises the intrinsic hierarchy
of memory scales encoded by the network, whereas the protocol
$\beta(t)$ determines how that hierarchy unfolds across macroscopic time.

\subsection{Null model ensembles}

Three null ensembles are used as counterfactual perturbations of each
empirical network. Each ensemble selectively disrupts a distinct structural
ingredient while preserving the others, allowing the contribution of
mesoscale organisation, weight geometry, and directionality to functional
memory organisation to be separated quantitatively. All observables were
computed from independent realizations of each ensemble.

\medskip
\noindent\textbf{Directed Configuration Model (DCM).}
The Directed Configuration Model preserves the in-degree and out-degree
sequences of every node while randomizing the global wiring pattern.
Realizations are generated through sequential edge rewiring until
convergence toward a uniform sample over directed graphs with the specified
degree sequences. Edge weights are preserved under rewiring.

Comparisons against the DCM isolate contributions arising from mesoscale
organisation beyond local degree constraints, including motif structure,
feedback cycles, hierarchical routing, and community organisation. Large
DCM Wasserstein deviations therefore indicate that the specific arrangement
of connections, rather than degree statistics alone, plays a dominant role
in organizing functional memory across depths. For each empirical network,
100 independent DCM realizations were generated.

\medskip
\noindent\textbf{Geometry Fluctuation Model (GFM).}
The Geometry Fluctuation Model preserves the binary topology exactly while
randomizing edge weights on the fixed directed skeleton. For
integer-valued weights, each nonzero entry $A_{ij}$ encodes the
multiplicity of parallel directed interactions from $j$ to $i$; changing
these multiplicities modifies the geometry of directed transport, including
weighted path structure, effective distances, and curvature-based
properties such as Forman--Ricci curvature~\cite{Forman2003,Weber2017},
while preserving the underlying binary graph. For real-valued weights,
$A_{ij}$ encodes a continuous interaction strength, and perturbation
modifies the geometric flow structure while preserving directed connectivity.

Each nonzero edge weight $A_{ij}$ is replaced by an independent random
draw preserving the mean: a zero-truncated Poisson variate with mean
$A_{ij}$ for integer-valued weights, and an exponential variate with mean
$A_{ij}$ for real-valued weights. In both cases, directed paths, strongly
connected components, and all unweighted topological observables are
preserved, whereas correlations between edge weights along directed
pathways are destroyed.

Comparisons against the GFM isolate the contribution of geometric flows to functional memory 
organisation while holding the binary
topology fixed. Large GFM deviations indicate that the specific geometric
organisation of interaction multiplicities sustains deep memory structure
beyond what can be explained by topology alone. For each empirical
network, 100 independent GFM realizations were generated.

\medskip
\noindent\textbf{Polarity Noise Model (PNM).}
The Polarity Noise Model perturbs edge directionality while preserving the
underlying interaction structure. Each directed edge is reversed
independently with probability $p$, while self-loops remain unchanged.

The parameter $p=0$ recovers the original network, whereas $p=0.5$
corresponds to maximal directional disorder, where each edge orientation is
an independent fair coin flip. The primary polarity analysis uses
\[
p\in\{0,0.01,0.1,0.2,0.3,0.5\},
\]
with 20 independent realizations generated at each polarity level.

The polarity response is quantified by
\begin{equation}
  \Delta_{\mathrm{pol}}
  =
  \NMR(p=0.5)-\NMR(p=0),
  \label{eq:delta_pol_methods}
\end{equation}
where $\NMR(p)$ denotes the ensemble-averaged non-Markovian ratio at
polarity level $p$. Negative values indicate that the original
directionality sustains deep memory organisation, because randomizing edge
directions reduces the non-Markovian ratio. Positive values indicate that
directionality suppresses deep memory organisation, because directional
randomization increases the non-Markovian ratio.

To quantify continuous sensitivity to directional perturbation, we also
compute the polarity slope
\[
s=\frac{\mathrm{d}\NMR}{\mathrm{d}p},
\]
estimated by ordinary least squares over the interval
$p\in[0,0.5]$. The polarity slope complements
$\Delta_{\mathrm{pol}}$ by measuring the rate at which memory organisation
responds to progressive directional disorder.

\medskip
\noindent\textbf{Ternary decomposition.}
To compare the relative contributions of mesoscale organisation, weight
geometry, and directionality, the Wasserstein deviations at reference depth
$l^\ast=20$ are normalized as
\begin{equation}
  \phi^{\mathrm{model}}
  =
  \frac{W_1^{\mathrm{model}}}{\Sigma},
  \qquad
  \Sigma
  =
  W_1^{\mathrm{DCM}}
  +
  W_1^{\mathrm{GFM}}
  +
  W_1^{\mathrm{PNM}},
  \label{eq:ternary_methods}
\end{equation}
yielding barycentric coordinates
\[
(\phi^{\mathrm{DCM}},
 \phi^{\mathrm{GFM}},
 \phi^{\mathrm{PNM}})
\]
within a ternary simplex. The three vertices correspond to idealized
dominance of mesoscale organisation, weight geometry, and directionality,
respectively.

The reference depth $l^* = 20$ was chosen as the deepest cascade level
that is uniformly available across all thirty-four corpus networks;
individual networks may have valid cascade information at greater depths,
but depth $20$ is the deepest level common to the full corpus and
therefore suitable for cross-network comparison.  It lies beyond the
shallow transient in which the HMS is dominated by short-range
compression events, and within the regime where the depth-resolved
Wasserstein curves vary smoothly and the cascade structure has stabilised
across networks.  The sensitivity of all cascade observables to
one-decade perturbations of the convergence threshold~$\varepsilon$ is
quantified in Supplementary Section 8; the
sensitivity of the derived classifications to nearby reference depths
is assessed in Supplementary Section 9.

Networks satisfying $\Sigma<0.01$ are classified as low-signal, indicating simultaneously weak separation
from all three null ensembles. No empirical network in the dataset
satisfied this criterion.

\subsection{Probabilistic classification of functional memory organisation}
\label{sec:classification}

To summarize the diversity of functional memory organisation across the
empirical networks, we performed an unsupervised Gaussian mixture analysis
on thermodynamic observables derived from the thermalization cascade. The
resulting latent components were interpreted along the two complementary
descriptive axes: the structural driver axis, identifying which null-model
perturbation produces the largest reorganisation of the memory hierarchy,
and the dynamical organisation axis, describing how memory is distributed,
compressed, and reorganised across depth scales.

Cluster assignments were obtained before semantic labels were applied. The
reported organisations therefore summarize emergent latent structure in the
thermodynamic observables rather than impose predefined categories on the
data.

\medskip
\noindent\textbf{Feature construction and model selection.}
The feature set combined observables describing memory depth, hierarchical
organisation, compression dynamics, null-model separation, and polarity
response. These included summaries of the effective memory depth,
hierarchical memory entropy and tail mass, compression-rate statistics,
Wasserstein deviations from the null ensembles, the polarity response
$\Delta_{\mathrm{pol}}$, the polarity slope $s$, and the ternary
coordinates
\[
(\phi^{\mathrm{DCM}},
 \phi^{\mathrm{GFM}},
 \phi^{\mathrm{PNM}}).
\]

All features were standardized by $z$-score normalization before fitting.
Gaussian mixture models were then fitted over a range of component numbers,
and the selected partitions minimized the Bayesian Information Criterion
(BIC). Bootstrap resampling with label alignment was used to estimate
assignment stability and posterior uncertainty.

\medskip
\noindent\textbf{Structural driver axis.}
Driver labels are assigned deterministically from the ternary coordinates
$(\phi^{\mathrm{DCM}}, \phi^{\mathrm{GFM}}, \phi^{\mathrm{PNM}})$
defined in equation~\eqref{eq:ternary_methods}, evaluated in the
following priority order.

\begin{enumerate}
\item \textbf{Dir} (directionality-dominated): $\phi^{\mathrm{PNM}}$
is the largest of the three coordinates.  A safety-net override
additionally reclassifies any network to Dir when $\phi^{\mathrm{PNM}}
\geq 0.45$ and $|s| \geq 0.06$, catching near-dominant PNM cases that
narrowly lose the argmax to $\phi^{\mathrm{DCM}}$ or $\phi^{\mathrm{GFM}}$
but carry a strong empirical polarity-response signal.

\item \textbf{Geo-Meso} (jointly geometry- and mesoscale-dominated):
$\phi^{\mathrm{PNM}} < 0.10$ and the relative DCM--GFM imbalance
\[
\frac{|\phi^{\mathrm{DCM}} - \phi^{\mathrm{GFM}}|}
     {\phi^{\mathrm{DCM}} + \phi^{\mathrm{GFM}}}
\;<\; 0.05,
\]
indicating that weighted transport geometry and mesoscale wiring
contribute comparably, with neither clearly dominant.

\item \textbf{Meso} (mesoscale-dominated): $\phi^{\mathrm{DCM}}$ is the
largest coordinate, indicating that higher-order wiring organisation
beyond the degree sequence is the primary structural contributor to the
memory hierarchy.

\item \textbf{Geo} (geometry-dominated): $\phi^{\mathrm{GFM}}$ is the
largest coordinate, indicating that weighted transport geometry is the
primary structural organiser.  Geometry-dominated networks are further
sub-ranked by the geometric dominance score $\phi^{\mathrm{GFM}} -
\phi^{\mathrm{DCM}}$ and divided into three equal tertiles:
Geo-H (highest dominance), Geo-M, and Geo-L, reflecting the degree to
which weight geometry dominates over mesoscale wiring within the
geometry-dominated sector.
\end{enumerate}

\medskip
\noindent\textbf{Dynamical organisation axis.}
The organisation labels are assigned post-hoc to the GMM clusters by
the following sequential criteria, applied to cluster-level summary
statistics (medians across cluster members).

\begin{itemize}

\item \textbf{Pulse}: the cluster with the highest median $\gamma_{\max}$,
provided it exceeds twice the next cluster's median $\gamma_{\max}$ or
lies above the $85$th percentile of the corpus-wide $\gamma_{\max}$
distribution.  This identifies systems whose compression profile is
dominated by a single extreme mode.

\item \textbf{Diffuse}: among the remaining clusters, the one with the
highest identification score $H_{\mathrm{clust}} + 0.20\,
\overline{\mathrm{NMR}}_{\mathrm{clust}}$, where $H_{\mathrm{clust}}$ is
the median normalised HMS entropy and $\overline{\mathrm{NMR}}$ is the
median non-Markovian ratio.  This selects the cluster with the broadest,
most high-entropy memory distribution.

\item \textbf{Reactive} and \textbf{Stable}: the two remaining
clusters are separated by the polarity reactivity score
\[
R_{\mathrm{clust}}
\;=\;
\mathrm{median}\bigl(\lvert s_i\rvert\bigr)
\;+\;
\mathrm{median}\bigl(\Delta^{\mathrm{range}}_i\bigr),
\]
where $s_i$ is the OLS slope of $\mathrm{NMR}(p)$ over $p\in[0,0.5]$
and $\Delta^{\mathrm{range}}_i = \max_p \mathrm{NMR}(p) -
\min_p \mathrm{NMR}(p)$ is the full polarity response range for network
$i$.  The cluster with the higher $R_{\mathrm{clust}}$ is tentatively
labelled Reactive.  A consistency check is then applied: if the
tentatively named Reactive cluster has a lower median $\gamma_{\max}$
than the tentatively named Stable cluster --- indicating that the
polarity score was inverted by an outlier network --- the two labels are
swapped.  This guard ensures that the label assignment is consistent with
the empirical separator between the two classes, which is the concentration
of compression dynamics ($\gamma_{\max}$) rather than the magnitude of
polarity response alone.

\end{itemize}

The resulting organisations reflect distinct dynamical signatures.
Diffuse systems exhibit high-entropy memory distributions spread across
many depth scales.  Reactive systems sustain deep memory with concentrated
compression dynamics and moderate directional sensitivity.  Stable systems
sustain deep memory with more diffuse compression profiles and weak
directional sensitivity.  Pulse systems exhibit compression concentrated
in a single dominant mode at one characteristic scale.  The final label
is therefore a combination of the initial polarity-score ordering and the
compression-dynamics consistency check, making it robust to outlier
networks whose polarity slopes are atypically large.

Full network assignments are reported in Supplementary Table~S1.

\bibliographystyle{naturemag}
\bibliography{../../../bib/biblio}

\pagebreak
{\centering\textbf{\Huge Supplementary Information}}

\renewcommand{\theequation}{S\arabic{equation}}
\renewcommand{\thefigure}{S\arabic{figure}}
\renewcommand{\thetable}{S\arabic{table}}
\renewcommand{\thesection}{S\arabic{section}}
\renewcommand{\thetheorem}{S\arabic{section}.\arabic{theorem}}
\makeatletter
\@ifundefined{theHsection}{}{\renewcommand{\theHsection}{supp.\arabic{section}}}
\@ifundefined{theHfigure}{}{\renewcommand{\theHfigure}{supp.\arabic{figure}}}
\@ifundefined{theHequation}{}{\renewcommand{\theHequation}{supp.\arabic{equation}}}
\makeatother
\setcounter{section}{0}
\setcounter{equation}{0}
\setcounter{figure}{0}
\setcounter{table}{0}
\setcounter{theorem}{0}


\section{Minimal example: memory compression in a single-node network}
\label{sec:single_node}

\begin{figure}[th]
  \centering
  \includegraphics[width=0.45\textwidth]{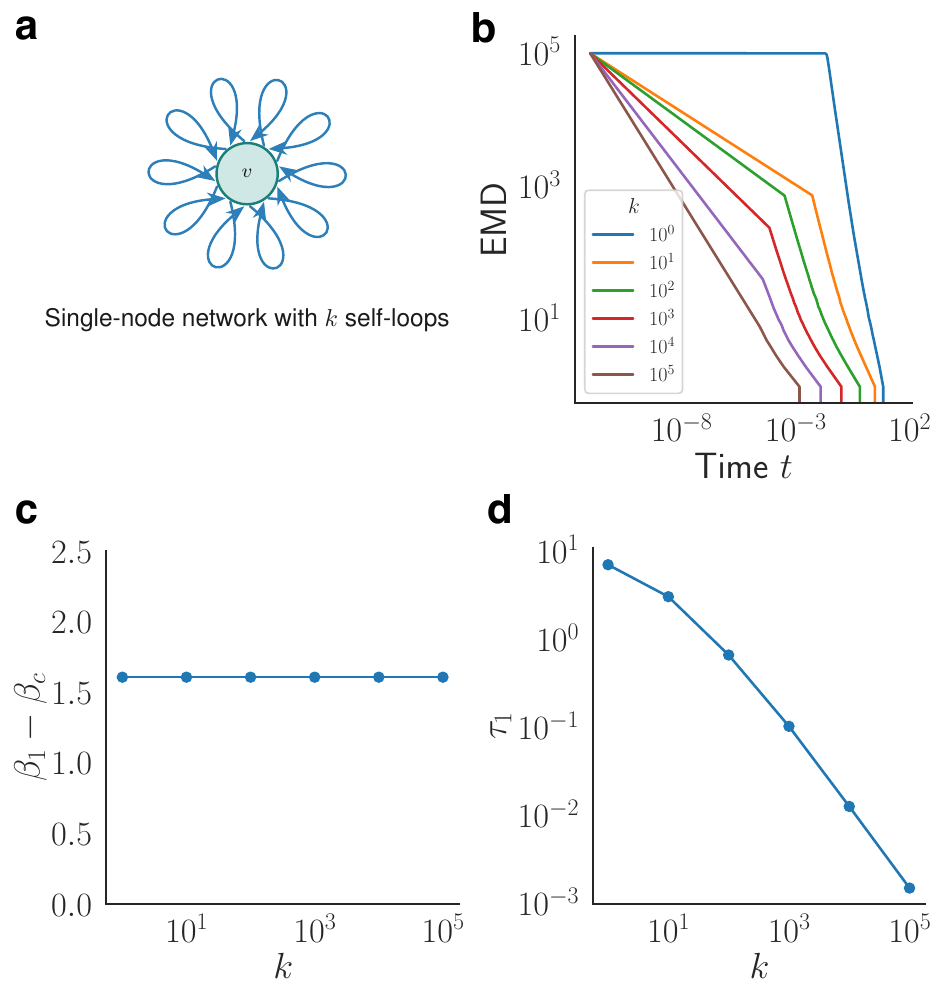}
  \caption{\textbf{Dynamic memory compression distinguishes systems with
  identical relative threshold structure.}
  \textbf{(a)}~Single-node network with $k$ self-loops, illustrating
  increasing feedback multiplicity without topological growth.
  \textbf{(b)}~Effective memory depth $\EMD(t)$ as a function of $t$ for
  different values of $k$; increasing loop multiplicity accelerates memory
  compression.
  \textbf{(c)}~The relative threshold $\beta_1-\beta_c$ is invariant with
  $k$, showing that the separation from criticality required for depth-one
  approximation is unchanged by feedback duplication.
  \textbf{(d)}~Thermalization time $\tau_1$ decreases with multiplicity,
  revealing macroscopic memory-compression dynamics not captured by the
  relative threshold alone.
  Panels~(b) and~(d) use log-log scale; panel~(c) uses log scale on the
  $x$-axis.}
  \label{fig:single-node}
\end{figure}

The distinction between equilibrium memory thresholds and macroscopic
memory-compression dynamics appears already in the simplest possible
system. Consider a single-node directed multigraph with $k$ parallel
self-loops (Fig.~\ref{fig:single-node}). Increasing $k$
increases feedback multiplicity without changing the binary topology.

For this system the adjacency matrix reduces to the scalar $A = [k]$, and
the infinite-memory resolvent is
\[
  R_\beta^{(\infty)}
  = \frac{1}{1-ke^{-\beta}},
  \qquad \beta > \beta_c = \log k.
\]
The finite-memory truncations are
\[
  R_\beta^{(n)}
  = \sum_{m=0}^{n}(ke^{-\beta})^m
  = \frac{1-(ke^{-\beta})^{n+1}}{1-ke^{-\beta}}.
\]
The relative truncation error is therefore
\[
  E_n(\beta)
  = \frac{|R_\beta^{(\infty)}-R_\beta^{(n)}|}{|R_\beta^{(\infty)}|}
  = (ke^{-\beta})^{n+1}
  = e^{-(n+1)(\beta-\beta_c)},
\]
from which the threshold inverse temperature satisfies
\[
  E_n(\beta_n) = \varepsilon
  \;\Longrightarrow\;
  \beta_n - \beta_c = \frac{-\log\varepsilon}{n+1}.
\]
In particular, $\beta_n - \beta_c$ depends only on $\varepsilon$ and $n$,
\emph{not} on $k$. The spectral gap $\beta_1-\beta_c = (-\log\varepsilon)/2$
is therefore invariant under changes in loop multiplicity, confirming the
invariance of the relative threshold illustrated in Fig.~\ref{fig:single-node}(c).

The macroscopic thermalization times do change substantially with $k$. Under
the protocol of Section~\ref{sec:twotimescale}, $\tau_n = \kappa(
e^{(\beta_n-\beta_c)^{1/\alpha}} - 1)$, which decreases as $k$ increases
through $\kappa = N/\bar{d} = 1/k$ (since the single-node network has
$N=1$ and mean out-degree $\bar{d}=k$). Systems with identical relative
equilibrium thresholds therefore exhibit distinct macroscopic organisations
of memory compression.

This example illustrates a general principle: equilibrium thresholds
characterise the intrinsic hierarchy of memory scales, whereas the
thermalization cascade describes how that hierarchy unfolds dynamically.
Functional memory organisation is not determined by static equilibrium
structure alone.

\section{Why the cascade is defined using resolvents rather than normalized propagators}
\label{sec:resolvent_vs_propagator}

The cascade is defined using the unnormalized resolvents $R_t^{(n)}$ rather
than the column-normalized propagators $T_t^{(n)}$. This choice is
essential. The resolvent records the total amplitude of directed walks up to
depth $n$, whereas normalization removes amplitude information by converting
each column into a probability distribution.

For a fixed source node $j$, the normalized propagator is
\[
  T_t^{(n)}(i,j)
  =
  \frac{R_t^{(n)}(i,j)}
       {\sum_{i'}R_t^{(n)}(i',j)}.
\]
Thus two resolvents with different total walk amplitudes but proportional
columns define the same propagator. The propagator retains the relative
distribution of flow targets but discards the magnitude of memory
amplification.

This loss is visible in the single-node example. For $A=[k]$, both the
finite and infinite resolvents are positive scalars, so
\[
  T_\beta^{(n)}=1
  \qquad\text{and}\qquad
  T_\beta^{(\infty)}=1
\]
for every $n$, every $k$, and every admissible $\beta$. If the cascade were
defined through the normalized propagators, every finite truncation would be
indistinguishable from the infinite-memory propagator, and no memory
hierarchy would be detected. 

More generally, normalization can collapse distinct memory amplitudes onto
identical transition probabilities whenever walk contributions remain
proportional after normalization. The propagator therefore describes the
distribution of effective flow, whereas the resolvent retains the
total accumulated contribution of long directed histories before
normalization.

By contrast, the resolvent error
\[
  \frac{\|R_\beta^{(\infty)}-R_\beta^{(n)}\|_F}
       {\|R_\beta^{(\infty)}\|_F}
\]
measures the amplitude of the unresolved Neumann tail: how much long-range
causal history remains before normalization turns the dynamics into a
transition matrix. The normalized propagators are used to visualize and
interpret effective flow, while the resolvents define thermalization because
they preserve the total memory amplitude that is progressively compressed by
the cascade.

\section{Monotonicity and normalization of the thermalization cascade}
\label{sec:monotonicity}

Let $A$ be a non-negative adjacency matrix with spectral radius $\rho(A)$.
We work in the regime $\beta>\beta_c=\log\rho(A)$, so that
$e^{-\beta}\rho(A)<1$ and the Neumann series
\[
  R_\beta^{(\infty)}
  =
  \sum_{k=0}^{\infty}e^{-\beta k}A^k
\]
converges entrywise and in any matrix norm.

\begin{lemma}[Monotonicity in memory depth]
\label{lem:monotone}
For every $\beta>\beta_c$,
\[
  R_\beta^{(\infty)}-R_\beta^{(n+1)}
  \;\leq\;
  R_\beta^{(\infty)}-R_\beta^{(n)}
  \qquad\text{entrywise,}
\]
and consequently $E_{n+1}(\beta)\leq E_n(\beta)$, where
\[
  E_n(\beta)
  =
  \frac{\|R_\beta^{(\infty)}-R_\beta^{(n)}\|_F}
       {\|R_\beta^{(\infty)}\|_F}.
\]
\end{lemma}

\begin{proof}
Since $A$ is non-negative, the tail sum
\[
  R_\beta^{(\infty)}-R_\beta^{(n)}
  =
  \sum_{k=n+1}^{\infty}e^{-\beta k}A^k
\]
is an entrywise non-negative matrix. Writing
\[
  R_\beta^{(\infty)}-R_\beta^{(n)}
  =
  e^{-\beta(n+1)}A^{n+1}
  +
  \bigl(R_\beta^{(\infty)}-R_\beta^{(n+1)}\bigr),
\]
and noting that $e^{-\beta(n+1)}A^{n+1}\geq 0$ entrywise, we conclude
\[
  R_\beta^{(\infty)}-R_\beta^{(n+1)}
  \;\leq\;
  R_\beta^{(\infty)}-R_\beta^{(n)}
  \qquad\text{entrywise.}
\]
Since the Frobenius norm is monotone on entrywise non-negative matrices,
dividing by $\|R_\beta^{(\infty)}\|_F>0$ gives $E_{n+1}(\beta)\leq
E_n(\beta)$.
\end{proof}

\begin{corollary}[Cascade ordering]
\label{cor:cascade_order}
For every $\varepsilon\in(0,1)$, the thermalization times satisfy
$\tau_0\geq \tau_1\geq \tau_2\geq\cdots$.
\end{corollary}

\begin{proof}
By Lemma~\ref{lem:monotone}, $E_{n+1}(t)\leq E_n(t)$ for all $t$, so
$\{t:E_{n+1}(t)\leq\varepsilon\}\supseteq\{t:E_n(t)\leq\varepsilon\}$.
Taking infima gives $\tau_{n+1}\leq\tau_n$.
\end{proof}

\begin{remark}
Corollary~\ref{cor:cascade_order} holds for any $\varepsilon\in(0,1)$
and is independent of the choice of thermodynamic protocol. The
hierarchical ordering of cascade times is an intrinsic property of the
network structure.
\end{remark}

\begin{proposition}[Normalization of the hierarchical memory spectrum]
\label{prop:hms_norm}
Assume the thermodynamic protocol is continuous, strictly increasing, and
satisfies $\beta(0)=\beta_c$. Then
\[
  \mu_G(l)
  =
  \frac{\tau_{l-1}-\tau_l}{\tau_0},
  \qquad l\geq1,
\]
defines a probability distribution: $\sum_{l\geq1}\mu_G(l)=1$.
\end{proposition}

\begin{proof}
Telescoping gives $\sum_{l=1}^{L}\mu_G(l) = (\tau_0-\tau_L)/\tau_0$,
so it suffices to show $\tau_L\to0$ as $L\to\infty$.

For every fixed $\beta>\beta_c$, the Neumann series converges absolutely,
hence $\|R_\beta^{(\infty)}-R_\beta^{(n)}\|_F\to0$ as $n\to\infty$.
Therefore, for every $\varepsilon>0$ and every $\beta>\beta_c$, there
exists $n_\beta$ such that $E_n(\beta)\leq\varepsilon$ for all
$n\geq n_\beta$. It follows that the threshold
\[
  \beta_n(\varepsilon)
  =
  \inf\{\beta>\beta_c:E_n(\beta)\leq\varepsilon\}
\]
satisfies $\limsup_{n\to\infty}\beta_n(\varepsilon)\leq\beta$ for every
$\beta>\beta_c$, and hence $\beta_n(\varepsilon)\to\beta_c$.

Since $\tau_n=\beta^{-1}(\beta_n)$ and $\beta^{-1}$ is continuous with
$\beta^{-1}(\beta_c)=0$ (by the protocol definition), we conclude
$\tau_n\to0$, and therefore $\sum_{l=1}^{L}\mu_G(l)\to1$.
\end{proof}

\section{Two-timescale structure and macroscopic rescaling}
\label{sec:twotimescale}

\subsection*{Two-timescale separation}

The framework operates on two distinct temporal levels.

\medskip\noindent\textbf{Microscopic level.}
At fixed macroscopic time $t$, the propagators $\Tn$ describe probability
flow through directed paths of depth at most $n$, weighted by
$e^{-\beta(t)k}$. These transitions are governed by the network adjacency
structure.

\medskip\noindent\textbf{Macroscopic level.}
The inverse-temperature parameter $\beta(t)$ evolves slowly according to
the thermodynamic protocol
\begin{equation}
  \beta(t) = \beta_c + \left[\log\!\left(1+\frac{t}{\kappa}\right)\right]^\alpha,
  \label{eq:protocol}
\end{equation}
with $\kappa = N/\bar{d}$ and $\alpha = 1 + 1/\log(1+\bar{d})$, where $N$
is the number of nodes and $\bar{d}$ the mean out-degree. The protocol is
strictly increasing with $\beta(0) = \beta_c$ and $\beta(t) \to \infty$ as
$t \to \infty$.

\medskip\noindent\textbf{Protocol inversion.}
Solving~\eqref{eq:protocol} for $t$:
\begin{equation}
  \beta^{-1}(\beta)
  = \kappa\!\left(e^{(\beta-\beta_c)^{1/\alpha}} - 1\right).
  \label{eq:protocol_inv}
\end{equation}
The thermalization times are therefore
\begin{equation}
  \tau_n = \kappa\!\left(e^{(\beta_n - \beta_c)^{1/\alpha}} - 1\right),
  \label{eq:tau_explicit}
\end{equation}
where $\beta_n(\varepsilon) = \inf\{\beta > \beta_c : E_n(\beta) \leq
\varepsilon\}$ depends only on $A$ and $\varepsilon$.

\subsection*{Invariance under macroscopic time rescaling}

\begin{proposition}[$\kappa$-invariance of normalized cascade observables]
\label{prop:invariance}
The non-Markovian ratio $\NMR = \tau_1/\tau_0$ and the hierarchical memory
spectrum $\{\mu_G(l)\}_{l \geq 1}$ are independent of the
sparsity-rescaling parameter $\kappa = N/\bar{d}$. Both depend only on the
adjacency matrix $A$ (through the intrinsic thresholds $\{\beta_n\}$ and
the degree-adapted exponent $\alpha$) and the tolerance $\varepsilon$.
\end{proposition}

\begin{proof}
From~\eqref{eq:tau_explicit},
\[
  \NMR
  = \frac{\tau_1}{\tau_0}
  = \frac{\kappa\bigl(e^{(\beta_1-\beta_c)^{1/\alpha}}-1\bigr)}
         {\kappa\bigl(e^{(\beta_0-\beta_c)^{1/\alpha}}-1\bigr)}
  = \frac{e^{(\beta_1-\beta_c)^{1/\alpha}}-1}
         {e^{(\beta_0-\beta_c)^{1/\alpha}}-1}.
\]
The factor $\kappa$ cancels exactly. Since $\beta_n$ depends only on
$(A,\varepsilon)$ and $\alpha$ depends only on $\bar{d}$ (determined by
$A$), the NMR is a function of $(A,\varepsilon)$ alone. For the HMS,
\[
  \mu_G(l)
  = \frac{\tau_{l-1}-\tau_l}{\tau_0}
  = \frac{\kappa\bigl(e^{(\beta_{l-1}-\beta_c)^{1/\alpha}}-e^{(\beta_l-\beta_c)^{1/\alpha}}\bigr)}
         {\kappa\bigl(e^{(\beta_0-\beta_c)^{1/\alpha}}-1\bigr)}.
\]
Again $\kappa$ cancels, so every spectral mass $\mu_G(l)$ depends only on
$(A,\varepsilon)$. \qed
\end{proof}

\begin{remark}[Role of $\alpha$ and cross-network comparison]
\label{rem:alpha}
Proposition~\ref{prop:invariance} establishes $\kappa$-invariance:
rescaling the overall macroscopic clock leaves NMR and HMS unchanged, so
networks of different sizes and densities are compared on genuinely
comparable thermodynamic footings. The exponent
$\alpha = 1 + 1/\log(1+\bar{d})$ does not cancel; it varies from
$\alpha\approx2.44$ for sparse networks ($\bar{d}\approx 1$) toward
$\alpha\to 1$ for dense networks ($\bar{d}\gg1$). The NMR therefore
carries a mild network-specific protocol dependence through $\alpha$.

This is a design feature rather than a defect: the $\alpha$ parameter
stretches the macroscopic clock in a degree-adapted manner, giving sparse
networks proportionally more time to explore deep cascade stages. For the
dataset analysed, the empirical variability of NMR across networks greatly
exceeds any variation attributable to $\alpha$ alone, as confirmed by the
sensitivity analysis in Supplementary Section~\ref{sec:eps_stability}.
\end{remark}

\begin{remark}[Ordering invariance under protocol reparameterization]
\label{rem:order_inv}
A stronger invariance holds for the qualitative cascade ordering. Let
$\tilde\beta(t)$ be any strictly increasing protocol with
$\tilde\beta(0)=\beta_c$. Then $\tilde\tau_n=\tilde\beta^{-1}(\beta_n)$
and $\beta_n<\beta_m \Leftrightarrow \tilde\tau_n<\tilde\tau_m$, since
$\tilde\beta^{-1}$ is strictly increasing. The ordering of thermalization
times — and hence the HMS support and the set of resolved cascade depths —
is therefore invariant under any strictly monotone reparameterization of
macroscopic time. The \emph{values} of NMR and HMS depend on the specific
protocol through $\alpha$, but their ordinal structure does not.
\end{remark}

\section{Weighted transport geometry}
\label{sec:geometry}

Edge weights $w_{j\to i}\in\mathbb{R}_{>0}$ encode the strength of
directed interactions and are treated as continuous positive reals
throughout this work (the integer-multiplicity interpretation is a
special case).  The weighted adjacency matrix has entries
$A_{ij}=w_{j\to i}$ when a directed edge $j\to i$ exists and $A_{ij}=0$
otherwise.

The matrix power $(A^n)_{ij}$ aggregates weighted directed walks of
length $n$ from $j$ to $i$, with each walk contributing the product of
its edge weights.  The infinite-memory transport operator
\begin{equation}
  R_\beta^{(\infty)}
  = \sum_{n=0}^{\infty}e^{-\beta n}A^n
  = (I - e^{-\beta}A)^{-1},
  \label{eq:Rbeta}
\end{equation}
is well-defined for $\beta$ exceeding the spectral radius of $A$ and
aggregates weighted transport contributions across all depths with
exponential suppression of long-range flow.  Changing edge weights
modifies path amplitudes, effective transport distances, and the
curvature geometry of directed flow.

\paragraph*{Forman--Ricci curvature for directed weighted networks.}

The Forman--Ricci curvature of a directed edge $e = (j\to i)$ in a
weighted directed graph is~\cite{Sreejith2016,Weber2017,Saucan2017}
\begin{equation}
  \mathrm{Ric}(j\to i)
  = \frac{w_j}{w_{j\to i}}
    + \frac{w_i}{w_{j\to i}}
    - \sum_{\substack{k\to j \\ k\neq i}}
        \frac{w_j}{\sqrt{w_{j\to i}\,w_{k\to j}}}
    - \sum_{\substack{i\to l \\ l\neq j}}
        \frac{w_i}{\sqrt{w_{j\to i}\,w_{i\to l}}},
  \label{eq:forman_ricci}
\end{equation}
where $w_{j\to i}=A_{ij}$ is the weight of the edge under
consideration, and $w_v$ is a positive weight assigned to node $v$
(set to unity if no node weights are available).

Three structural features of Eq.~\eqref{eq:forman_ricci} deserve
emphasis.  First, the sum over $k\to j$ runs over all edges
\emph{incoming to the tail} $j$ (excluding the reverse edge $i\to j$
if present); the sum over $i\to l$ runs over all edges \emph{outgoing
from the head} $i$ (excluding $j\to i$).  This direction-compatible
neighbor selection is the essential modification from the undirected
case~\cite{Sreejith2016}: only neighbors whose edge direction is
consistent with a continued directed path through $e$ contribute to
the curvature.  Second, each neighbor term carries the node weight of
the shared endpoint in the numerator, so curvature is sensitive to the
relative importance of nodes as well as the relative strength of edges.
Third, the formula contains no cross-term between the tail neighborhood
and the head neighborhood; such a term appears in Forman's original
CW-complex formula~\cite{Forman2003} but vanishes for one-dimensional
cell complexes (graphs).

In the unweighted combinatorial case ($w_v=1$ for all $v$, $w_e=1$ for
all $e$), Eq.~\eqref{eq:forman_ricci} reduces to
\begin{equation}
  \mathrm{Ric}(j\to i)
  = 2 - \deg^{\mathrm{in}}(j) - \deg^{\mathrm{out}}(i),
  \label{eq:forman_combinatorial}
\end{equation}
where $\deg^{\mathrm{in}}(j)$ is the in-degree of the tail and
$\deg^{\mathrm{out}}(i)$ is the out-degree of the head.  An edge
connecting a high-in-degree tail to a high-out-degree head therefore
has strongly negative curvature, signaling a bottleneck or hub-to-hub
connection; edges at the periphery of a sparse network have curvature
near $+2$.

\paragraph*{Effect of weight randomization.}

Replacing each $w_{j\to i} = A_{ij}$ with an independent draw
$\tilde{w}_{j\to i}$ from a distribution supported on $\mathbb{R}_{>0}$
with mean $A_{ij}$ (for example, a Gamma distribution matched to the
mean and variance of the empirical weight distribution, or a
zero-truncated log-normal) modifies every curvature term in
Eq.~\eqref{eq:forman_ricci} that involves this edge:
\begin{itemize}
  \item the leading $w_j/w_{j\to i}$ and $w_i/w_{j\to i}$ terms
        (edge weight in the denominator),
  \item every neighbor term $1/\sqrt{w_{j\to i}\,w_{k\to j}}$ or
        $1/\sqrt{w_{j\to i}\,w_{i\to l}}$ that shares the edge
        $j\to i$ (edge weight appears under the square root).
\end{itemize}
Because adjacent edges share endpoint nodes, randomizing a single edge
perturbs the curvature of all neighboring edges as well.  The GFM
(graph-filtered model) therefore simultaneously randomizes both local
curvature and its spatial correlations while preserving the binary
skeleton (which edges exist) and hence directed reachability and
unweighted path structure.

Deviations of the empirical curvature distribution from that of the
GFM ensemble isolate the contribution of \emph{structured weighted
transport geometry} — systematic covariation of weights along directed
pathways — to functional network organisation.  In particular, directed
paths along which edge weights are positively correlated (a
``directed highway'') exhibit systematically more positive curvature
than expected under the GFM, while bottleneck edges embedded between
dense neighborhoods exhibit more negative curvature.  These signatures
are destroyed by independent weight randomization, making the GFM
comparison a sensitive probe of coordinated weight structure.

\section{Wasserstein geometry of hierarchical memory spectra}
\label{sec:wasserstein}

The hierarchical memory spectrum $\mu_G(l) = (\tau_{l-1}-\tau_l)/\tau_0$
defines a probability distribution over ordered memory depths. Because
memory scales possess a natural hierarchical order, comparing spectra
requires a metric that respects this structure.

For probability distributions $\mu$ and $\nu$ on $\{1,\ldots,l_{\max}\}$,
the 1-Wasserstein distance is
\begin{equation}
  W_1(\mu,\nu)
  = \sum_{l=1}^{l_{\max}}
    \left|\sum_{k=1}^{l}\mu(k) - \sum_{k=1}^{l}\nu(k)\right|,
  \label{eq:wasserstein}
\end{equation}
equal to the $L^1$ distance between cumulative distribution functions. This
admits the optimal-transport interpretation: $W_1(\mu,\nu)$ is the minimal
cost of redistributing unit mass from $\mu$ to $\nu$ when moving mass
between depths $i$ and $j$ costs $|i-j|$.

The support $l_{\max}$ is set per-network to the deepest resolved cascade
level (the largest $n$ for which $\tau_n>0$ at the given $\varepsilon$).
For $l>l_{\max}$, both $\mu(l)$ and $\nu(l)$ are zero, so the cumulative
functions agree and the additional terms in~\eqref{eq:wasserstein} vanish.
The distance is therefore insensitive to the precise choice of $l_{\max}$
beyond the true cascade depth.

\begin{proposition}[Wasserstein vs.\ $L^2$ comparison]
\label{prop:w1_l2}
If $\mu=\delta_{l_1}$ and $\nu=\delta_{l_2}$ are point masses at depths
$l_1\neq l_2$, then $W_1(\mu,\nu)=|l_1-l_2|$ whereas
$\|\mu-\nu\|_2=\sqrt{2}$ regardless of $|l_1-l_2|$. The Wasserstein
metric is therefore sensitive to the spatial displacement of memory mass
along the depth axis, whereas the $L^2$ norm is not.
\end{proposition}

\begin{proof}
For $\mu=\delta_{l_1}$ and $\nu=\delta_{l_2}$, the cumulative functions
differ by $1$ on $[\min(l_1,l_2),\max(l_1,l_2)-1]$ and agree elsewhere,
giving $W_1=|l_1-l_2|$. The $L^2$ norm equals
$\sqrt{(\mu(l_1)-\nu(l_1))^2+(\mu(l_2)-\nu(l_2))^2}=\sqrt{2}$ for all
$l_1\neq l_2$.
\end{proof}

\section{Null-model comparisons and polarity response trajectories}
\label{sec:nmr_nullmodels}

\paragraph{Null-model NMR ratios.}
For each of the thirty-four corpus networks, the non-Markovian ratio of
the real network was compared against bootstrap null ensembles for all
three structural null models: the Geometry Fluctuation Model (GFM), the
Directed Configuration Model (DCM), and the Polarity Noise Model at
$p=0.5$ (PNM).  Figure~\ref{fig:nmr_overview}b shows the ratio
$\NMR_{\mathrm{model}}/\NMR_{\mathrm{real}}$ for all thirty-four networks
and all three models.  Each point is the bootstrap mean over $n=100$
randomisations (GFM and DCM) or $n=20$ (PNM); error bars are $95\%$
bootstrap confidence intervals.

The three null models reveal qualitatively distinct structural roles.
For the GFM, the ratio lies strictly below $1.0$ in thirty-three of the
thirty-four networks, and the $95\%$ confidence interval excludes $1.0$ in
thirty-two of them, establishing that weighted transport geometry
systematically deepens functional memory across every domain studied.
The one marginal exception (Chesapeake Bay wet-season, ratio
$1.039 \pm \mathrm{CI}$) reflects model-estimated interaction weights as
discussed in the main text.  The GFM panel therefore displays the
strongest and most consistent signal of the three null models: ratios span
$0.704$--$1.039$ with a median of $0.880$, and no network lies
significantly above $1.0$.

For the DCM, ratios scatter bidirectionally around $1.0$ ($18$ networks
with $\NMR_{\mathrm{real}} > \NMR_{\mathrm{DCM}}$, $15$ with
$\NMR_{\mathrm{real}} < \NMR_{\mathrm{DCM}}$, $1$ overlap), confirming
that mesoscale wiring organisation can either facilitate or suppress deep
memory relative to the degree-sequence expectation, with no universal
direction.  For PNM at $p=0.5$, the split is similarly bidirectional
($15$ above, $14$ below, $5$ overlap), indicating that partial polarity
randomisation neither universally deepens nor universally suppresses
memory across the corpus.

\paragraph{Polarity trajectory overview.}
Figure~\ref{fig:nmr_overview}a shows the normalised polarity trajectories
$\NMR_{\mathrm{null}}(p)/\NMR_{\mathrm{real}}$ for all thirty-four
networks, grouped by dynamical organisation.  Each trajectory begins at
$1.0$ by construction ($p=0$, no polarity perturbation) and is evaluated
at $p \in \{0.01, 0.1, 0.2, 0.3, 0.5, 0.6, 0.75, 0.9, 1.0\}$.

Trajectories within the Diffuse organisation form a tight cluster with
minimal departure from $1.0$ throughout the polarity sweep, reflecting the
weak and sign-heterogeneous directional sensitivity characteristic of
this class.  The Reactive organisation shows moderate spread with both
upward and downward excursions, consistent with mixed but non-negligible
directional responses.  The Stable organisation is also broadly clustered
near $1.0$, but includes one prominent outlier: the JDK software dependency
graph reaches a normalised ratio of $1.70$ at $p=0.5$ before returning
toward $1.0$ at $p=1.0$, making it the strongest polarity responder in
the entire dataset.  The Pulse organisation displays the clearest
divergence: US air traffic is effectively flat (dynamically inert under
polarity perturbation), whereas plant Photosystem~I shows a pronounced
decrease reaching a ratio of $0.81$ at $p=0.5$, reflecting extreme
sensitivity of its single-mode cascade to directional disruption.  The
near-recovery of all trajectories toward $1.0$ at $p=1.0$ (full reversal)
confirms that polarity perturbation is approximately symmetric around the
halfway point for most networks.

\paragraph{Per-network polarity trajectories.}
Figure~\ref{fig:nmr_trajectories} displays the full NMR$(p)$ trajectory
for each of the thirty-four corpus networks individually.  Each panel
shows the absolute NMR as a function of polarity $p$, with shaded $95\%$
bootstrap confidence bands ($n=20$ resamples per polarity level).  Two
horizontal dashed reference lines indicate the DCM and PNM null-model NMR
means for each network, providing a per-network scale reference for the
magnitude of the polarity response relative to the structural null baselines.
Networks are coloured by domain.  The panels collectively show that polarity
response shapes are network-specific and span a wide range of trajectory
forms, from the approximately invariant flat trajectories of US air traffic
and \textit{E.~coli} to the strongly peaked trajectories of JDK
dependencies and bitcoin trust, and the strongly concave trajectory of
plant Photosystem~I.

\begin{figure*}[!h]
  \centering
  \includegraphics[width=\textwidth]{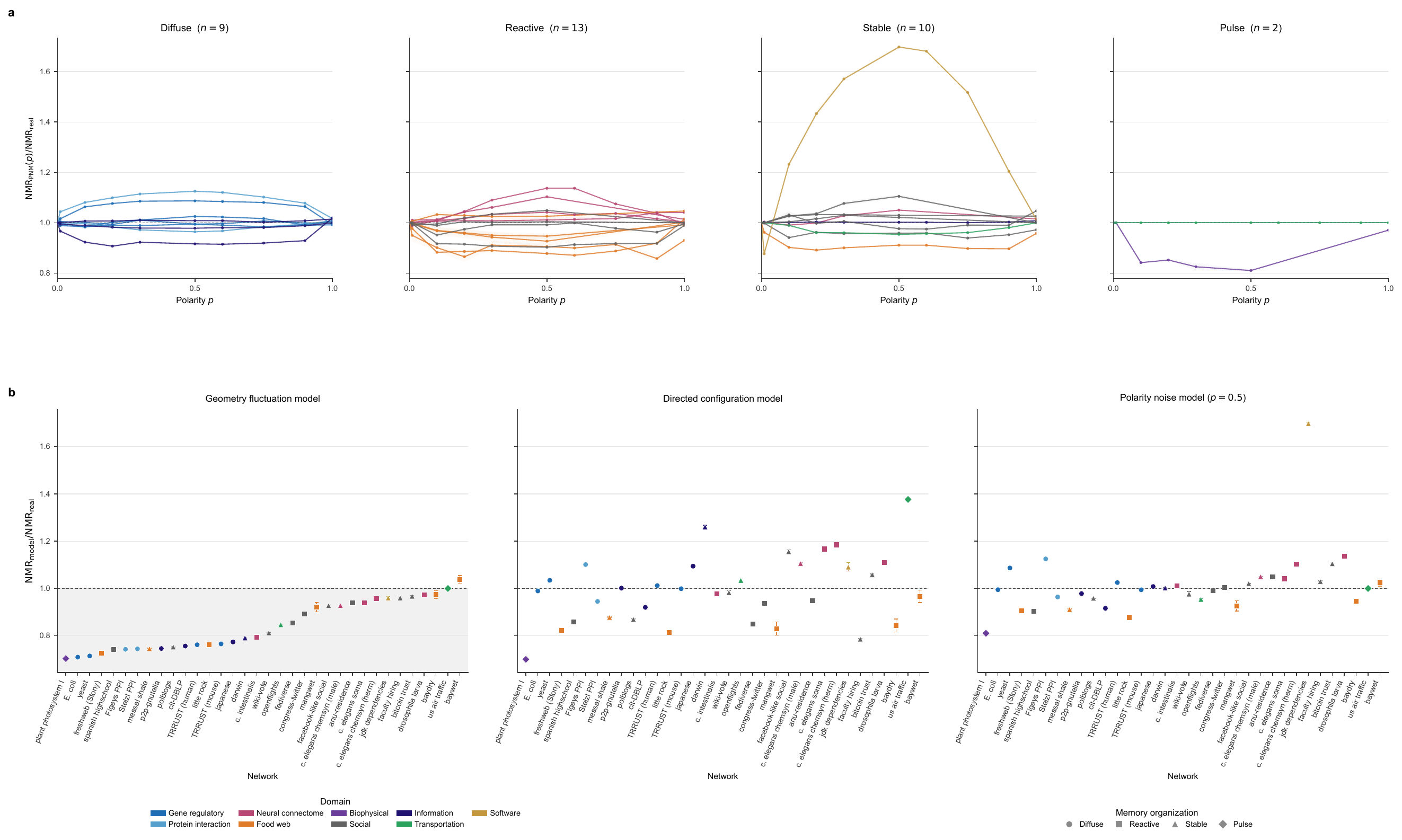}
  \caption{
    \textbf{Null-model NMR comparisons and polarity trajectory overview
    for all thirty-four corpus networks.}
    \textbf{a}, Normalised polarity trajectories grouped by dynamical
    organisation.  Each line shows $\NMR_{\mathrm{null}}(p)/\NMR_{\mathrm{real}}$
    as a function of polarity perturbation level $p$ for one
    network, colour-coded by domain.  All trajectories are anchored at
    $1.0$ at $p=0$ (no perturbation) by construction.  The four panels
    correspond to the Diffuse ($n=9$), Reactive ($n=13$), Stable ($n=10$),
    and Pulse ($n=2$) dynamical organisations identified in the main
    analysis.
    \textbf{b}, Null-model NMR ratios $\NMR_{\mathrm{model}}/\NMR_{\mathrm{real}}$
    for all thirty-four networks under each of the three null models:
    Geometry Fluctuation Model (GFM, left), Directed Configuration Model
    (DCM, centre), and Polarity Noise Model at $p=0.5$ (PNM, right).
    Networks are ordered along the horizontal axis by ascending ratio
    within each panel.  Points show bootstrap means; error bars indicate
    $95\%$ confidence intervals ($n=100$ resamples for GFM and DCM;
    $n=20$ for PNM).  The dashed horizontal line marks the ratio of $1.0$;
    ratios below $1.0$ (shaded region) indicate that the real network
    sustains deeper memory than the null ensemble.  Marker shape indicates
    dynamical organisation; marker colour indicates domain.  The GFM panel
    shows a near-universal signal: thirty-two of thirty-four networks lie
    significantly below $1.0$, with one marginal exception (Chesapeake Bay
    wet-season food web) and one exactly at $1.0$ (US air traffic, Pulse).
    The DCM and PNM panels show bidirectional scatter around $1.0$,
    confirming that mesoscale wiring and directionality lack a universal
    directional effect on memory depth.
  }
  \label{fig:nmr_overview}
\end{figure*}

\begin{figure*}[!h]
  \centering
  \includegraphics[width=\textwidth]{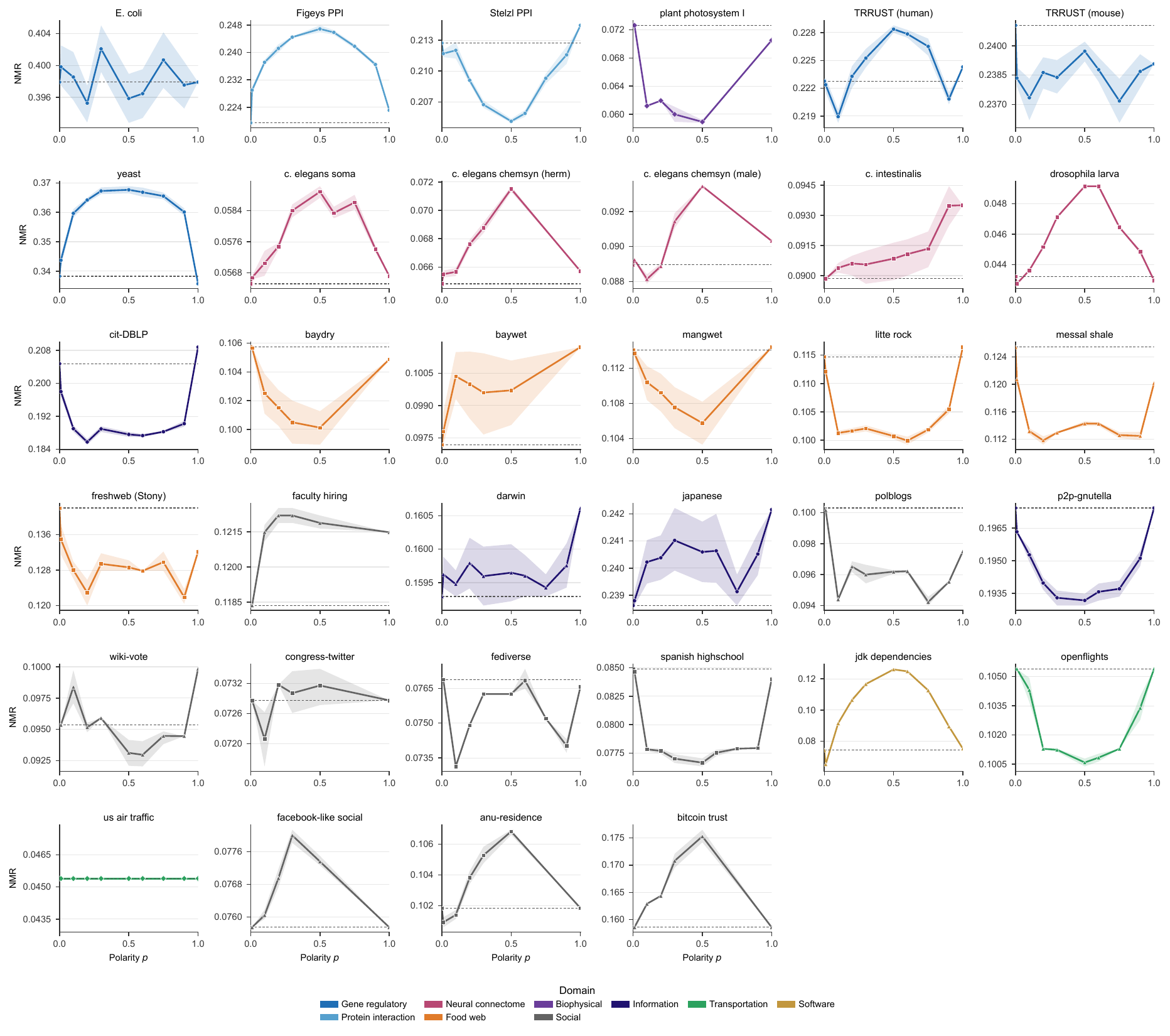}
  \caption{
    \textbf{Per-network polarity response trajectories for all
    thirty-four corpus networks.}
    Each panel shows $\NMR(p)$ as a function of the polarity perturbation
    level $p$ for one network.  Solid lines connect bootstrap means;
    shaded bands show $95\%$ confidence intervals ($n=20$ resamples per
    polarity level). Horizontal dashed lines indicate the DCM null-model mean NMR
    (lower) and the PNM null-model mean NMR at $p=0.5$ (upper) for
    each network, providing structural reference baselines on the
    per-network scale.  Marker colour denotes domain.  The panels are
    arranged alphabetically.  Trajectory shapes
    range from effectively flat (US air traffic, Meso/Pulse) through
    weakly modulated (most Diffuse networks) to strongly peaked (JDK
    dependencies, Stable/Dir, ratio $1.70$ at $p=0.5$) and strongly
    concave (plant Photosystem~I, Geo/Pulse, ratio $0.81$ at $p=0.5$).
  }
  \label{fig:nmr_trajectories}
\end{figure*}

\section{Tolerance stability of the hierarchical memory organisation}
\label{sec:eps_stability}

\paragraph{What $\varepsilon$ controls.}
Every observable used in the main analysis (the hierarchical memory
spectrum (HMS), the non-Markovian ratio, Wasserstein null-model deviations,
and compression profiles) is derived from the thermalization cascade,
whose individual steps are certified by a convergence threshold~$\varepsilon$
(Methods, Eq.~\eqref{eq:tau_methods}).  We assessed whether the cascade is
locally stable around the reference value $\varepsilon_{\mathrm{ref}} =
10^{-5}$ by comparing the HMS at the two adjacent tolerances $\varepsilon
\in \{10^{-6},\,10^{-4}\}$ against the reference spectrum.

\paragraph{Metric.}
For each network $i$ and each local tolerance $\varepsilon$, we compute
the depth-averaged Wasserstein displacement
\[
\langle W_1(n)\rangle_n
\;=\;
\frac{1}{n_{\max}}\sum_{n=1}^{n_{\max}}
W_1\!\bigl(\mu_{\varepsilon}^{(n)},\,\mu_{10^{-5}}^{(n)}\bigr),
\]
where $\mu_\varepsilon^{(n)}$ is the HMS at tolerance~$\varepsilon$ truncated
at memory depth~$n$.  The local instability score for each network is then
\[
\delta_i
\;=\;
\mathrm{median}_{\varepsilon\in\{10^{-6},\,10^{-4}\}}
\langle W_1(n)\rangle_n,
\]
the median displacement across the two local perturbation directions.
Small $\delta_i$ confirms that the cascade is insensitive to whether the
convergence threshold is tightened or relaxed by one order of magnitude.
All thirty-four corpus networks produced valid HMS data at both local
tolerances.

\begin{figure*}[!h]
  \centering
  \includegraphics[width=\textwidth]{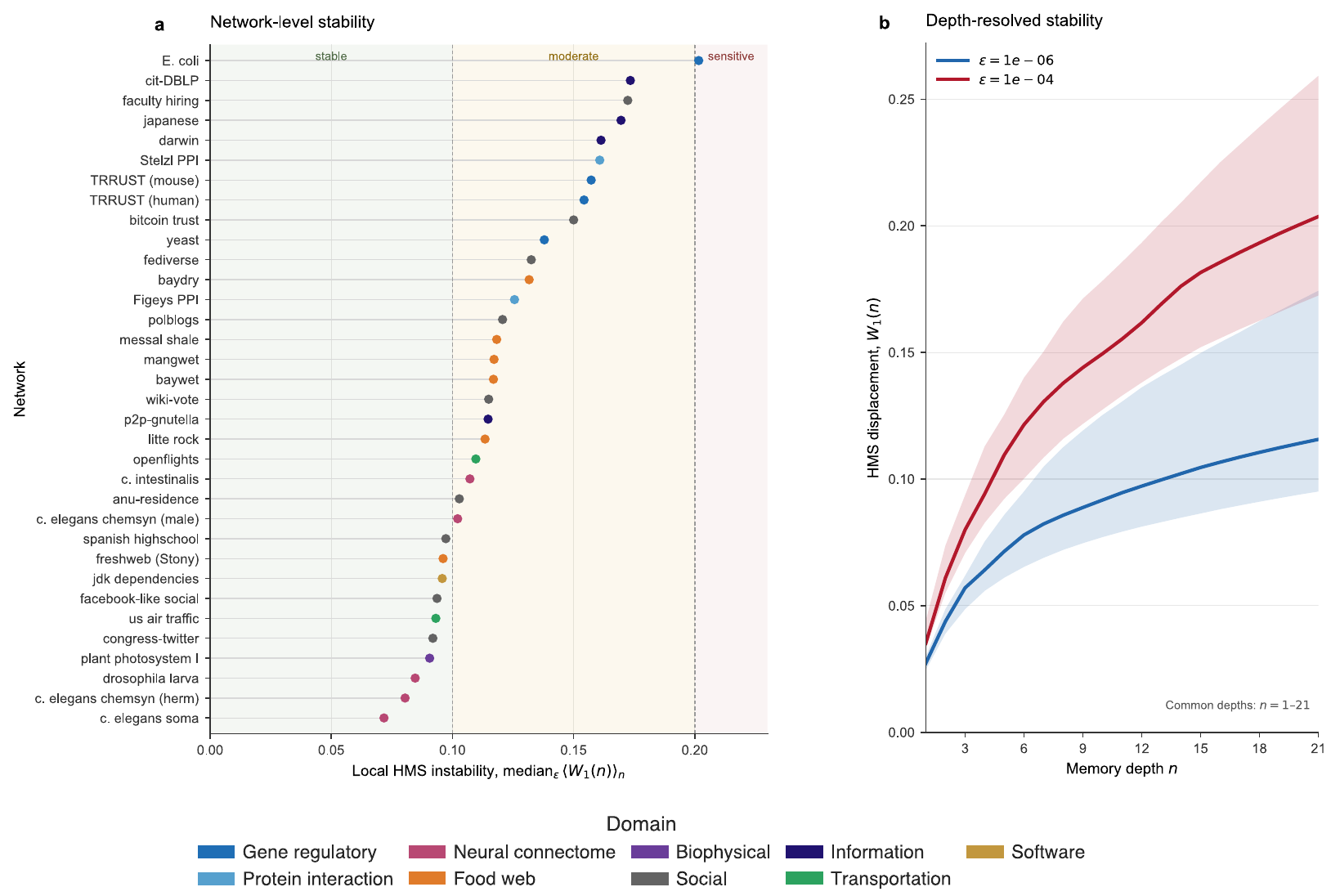}
  \caption{
    \textbf{Local numerical stability of the hierarchical memory cascade
    around the reference tolerance $\varepsilon_{\mathrm{ref}} = 10^{-5}$.}
    \textbf{a}, Ranked network-level instability scores.  For each network,
    the local instability $\delta_i = \mathrm{median}_\varepsilon\langle
    W_1(n)\rangle_n$ is the median depth-averaged 1-Wasserstein displacement
    between the hierarchical memory spectrum at tolerance $\varepsilon$ and
    the reference spectrum at $\varepsilon = 10^{-5}$, taken over the two
    adjacent tolerances $\varepsilon \in \{10^{-6}, 10^{-4}\}$.  Vertical
    dashed lines mark the locally stable ($\delta_i < 0.10$), moderately
    stable ($0.10 \leq \delta_i < 0.20$), and locally sensitive ($\delta_i
    \geq 0.20$) zones.  All thirty-four networks produce valid cascades at
    both local tolerances.  Marker colour denotes domain.
    \textbf{b}, Depth-resolved displacement profiles.  Median $W_1(n)$
    across all thirty-four networks as a function of memory depth~$n$, for
    the two local tolerances (blue: $\varepsilon = 10^{-6}$; red:
    $\varepsilon = 10^{-4}$).  Shaded bands indicate the interquartile
    range across networks.  Displacements are computed over the twenty-one
    common memory depths available for all networks at both local tolerances
    ($n = 1$--$21$).  The monotonic growth with depth reflects increasing
    sensitivity of deep cascade contributions to the convergence threshold,
    consistent with the exponential suppression factor $e^{-\beta n}$ in the
    resolvent.  The consistently lower blue curve confirms that
    $\varepsilon_{\mathrm{ref}} = 10^{-5}$ is positioned at the accurate
    end of the locally stable precision range.
  }
  \label{fig:stability}
\end{figure*}

\paragraph{Network-level results.}
Figure~\ref{fig:stability}a ranks the thirty-four networks by their local
instability score~$\delta_i$.  The scores span $0.077$ (somatic
\textit{C.~elegans} connectome) to $0.258$ (\textit{E.~coli}
transcriptional network), with a median of $0.128$ across the corpus.
Seven networks are locally stable ($\delta_i < 0.10$), twenty-five are
moderately stable ($0.10 \leq \delta_i < 0.20$), and two marginally
exceed the sensitive threshold ($\delta_i \geq 0.20$): the
\textit{E.~coli} transcriptional network ($\delta_i = 0.258$) and the
DBLP citation network ($\delta_i = 0.204$).  Both exceedances are
marginal: they reflect the comparatively large and narrow HMS spectra of
these two networks, which makes their depth-averaged Wasserstein distance
more sensitive to the convergence threshold than networks with broader
spectral distributions.  In neither case does the HMS shift alter the
qualitative organisation of memory across depth scales.

\paragraph{Depth-resolved results and direction asymmetry.}
Figure~\ref{fig:stability}b shows the median Wasserstein displacement
$W_1(n)$ as a function of memory depth across all thirty-four networks.
At both local tolerances, the displacement grows monotonically with~$n$:
at depth $n=1$, median $W_1 = 0.027$ ($\varepsilon = 10^{-6}$) and
$0.035$ ($\varepsilon = 10^{-4}$); at depth $n=21$, median $W_1 = 0.116$
and $0.204$ respectively.  This monotonic growth is physically expected
from the resolvent structure: deep cascade contributions carry
an exponential suppression factor $e^{-\beta n}$ and therefore require
finer threshold resolution to certify precisely.

The displacement at $\varepsilon = 10^{-4}$ (red curve in
Fig.~\ref{fig:stability}b) exceeds that at $\varepsilon = 10^{-6}$ (blue
curve) by a factor of $1.65$ in median.  The cascade is therefore more
stable in the finer direction than in the coarser direction, establishing
that $\varepsilon_{\mathrm{ref}} = 10^{-5}$ sits at the accurate end of
the locally stable precision range.

\paragraph{Summary.}
Thirty-two of the thirty-four corpus networks fall within the moderately
stable or locally stable range ($\delta_i < 0.20$).  The two marginal
exceedances ($E.~coli$ and cit-DBLP) reflect spectral concentration rather
than numerical instability and do not affect the qualitative conclusions
drawn for those networks.  The reference tolerance $\varepsilon = 10^{-5}$
is positioned at the finest precision at which the full cascade is
certifiable across the corpus, and the local stability analysis confirms
that the cascade observables used throughout the main analysis are robust
to one-decade perturbations of this parameter.

\section{Robustness of the unsupervised classification to HMS depth}
\label{sec:depth_robustness}

\paragraph{Setup.}
The unsupervised classification operates on the hierarchical memory
spectrum (HMS) truncated at a maximum memory depth~$n$.  We assessed
whether the classification is stable with respect to this truncation by
repeating the full classification pipeline (feature extraction, GMM
fitting, BIC-based model selection, and post-hoc label assignment) at
five reduced depths $n \in \{5, 8, 11, 14, 17\}$ and comparing the
resulting labels and posterior probabilities against the reference
classification obtained at the full cascade depth $n = 20$.  The
reference depth was chosen as the deepest cascade level available across
all thirty-four networks; shorter depths truncate the HMS at earlier
compression stages.  At each depth, all thirty-four networks are
classified independently and the resulting regime and driver labels are
aligned to the reference using the Hungarian matching algorithm before
computing agreement statistics.

\begin{figure*}[t]
  \centering
  \includegraphics[width=\textwidth]{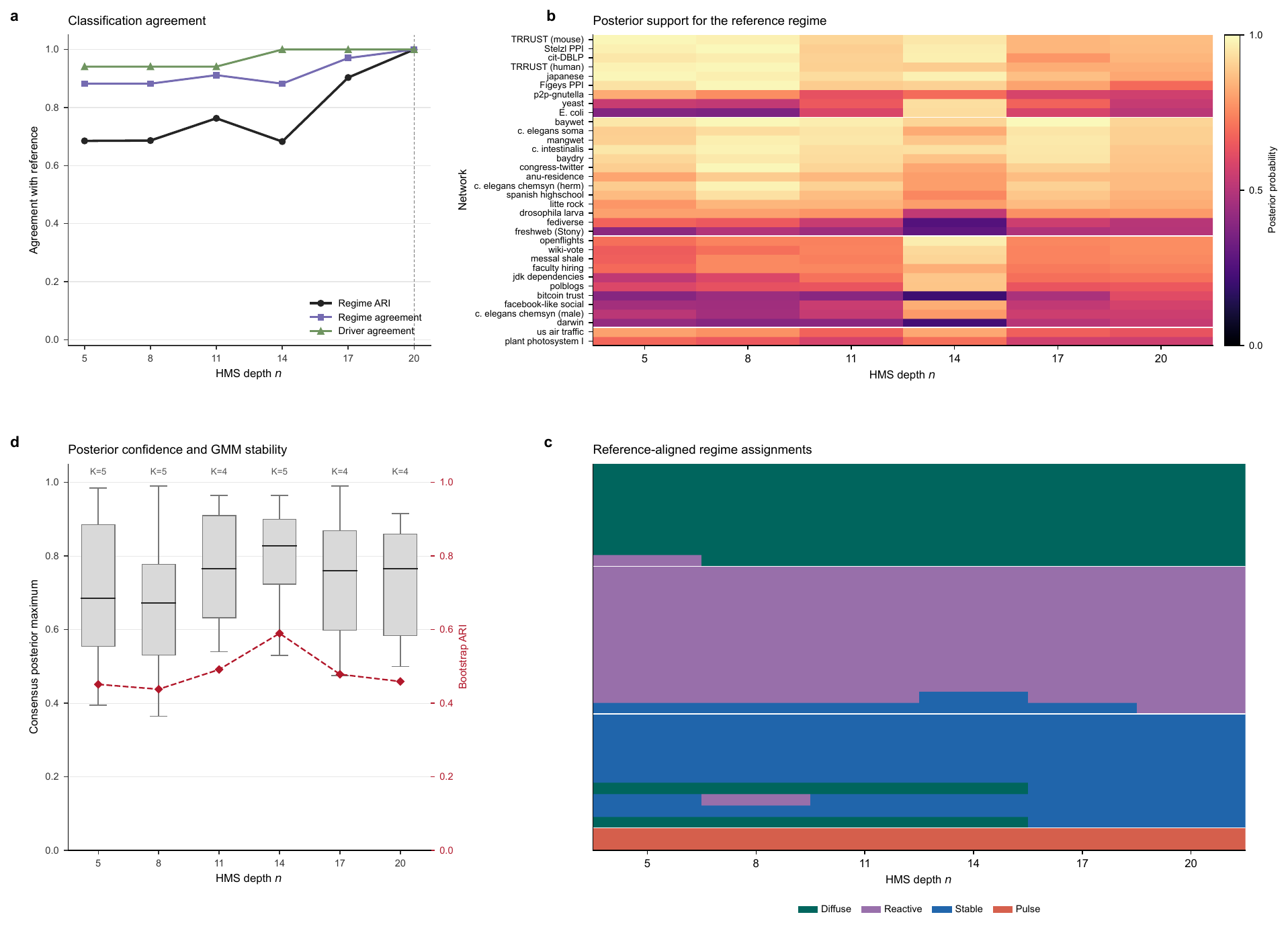}
  \caption{\textbf{Robustness of the unsupervised classification to the hierarchical memory depth.}
\textbf{a,} Agreement between classifications obtained at HMS depths $n\in\{5,8,11,14,17,20\}$ 
and the reference classification at $n=20$. Black circles show the 
adjusted Rand index (ARI) between regime partitions, purple squares 
show the fraction of networks retaining their reference regime, and 
green triangles show driver-label agreement. 
\textbf{b,} Posterior probability assigned to each network's reference 
regime across depth. Networks are ordered by their regime at $n=20$; 
lighter colours indicate stronger posterior support.
\textbf{c,} Reference-aligned regime assignments across depth. 
Depth-specific mixture components were aligned to the reference 
regimes by maximum posterior overlap, permitting multiple components 
to map to the same regime when the Bayesian information criterion 
selected a larger number of components. Colours denote Diffuse, Reactive, 
Stable and Pulse memory organisations.
\textbf{d,} Distribution of the maximum bootstrap-consensus regime 
posterior across networks at each depth. Boxes show the interquartile 
range, central lines indicate medians and whiskers show the non-outlier 
range. Red diamonds denote the bootstrap ARI of the fitted Gaussian 
mixture model, and labels above the boxes give the number of mixture 
components selected by the Bayesian information criterion. Regime 
agreement remained between $88\%$ and $97\%$, while driver agreement 
remained between $94\%$ and $100\%$ across the tested depths, indicating 
that the inferred functional-memory organisation is robust to the 
selected HMS depth.}
  \label{fig:depth_robustness}
\end{figure*}

\paragraph{Driver axis.}
The ternary driver coordinates $(\phi^{\mathrm{DCM}}, \phi^{\mathrm{GFM}},
\phi^{\mathrm{PNM}})$ and their derived labels converge rapidly with
cascade depth (Fig.~\ref{fig:depth_robustness}a,b).  Driver agreement
with the reference reaches $100\%$ at $n = 14$ and holds at all deeper
depths.  At $n = 5$, $94.1\%$ of networks receive the correct driver
label; the two disagreements are within-Geo sub-rank swaps
(\textit{E.~coli} Geo-M\,$\to$\,Geo-H and one other network Geo-H\,$\to$\,Geo-M),
with no network crossing a category boundary (Geo\,$\to$\,Meso, Meso\,$\to$\,Dir,
or similar).  The median ternary drift (the $L_1$ distance between the
normalised Wasserstein fractions at depth~$n$ and the reference)
decreases monotonically from $0.054$ at $n = 5$ to $0.005$ at $n = 17$,
confirming that the driver coordinates converge well before regime
classification reaches its reference state.

\paragraph{Regime axis: high-confidence networks.}
The twenty-two networks whose reference classification carries a consensus
posterior maximum above $0.65$ agree with the reference regime label at
every tested depth without exception (Fig.~\ref{fig:depth_robustness}b).
This complete agreement across all five reduced depths (including the
most severely truncated cascade at $n = 5$) demonstrates that the
core of the classification is robust to substantial reduction in cascade
information.  The identity of these networks spans all four dynamical
organisations and all six domains in the corpus.

\paragraph{Regime axis: borderline networks.}
The twelve networks with reference posterior maximum below $0.65$ show
varying degrees of instability.  Six networks disagree with the reference
at one or more depths; all six have reference posterior maxima between
$0.500$ and $0.615$.  The three most persistent disagreers are freshweb
(Stony, pmax\,$=\,0.500$, disagrees at all five reduced depths), darwin
(pmax\,$=\,0.540$, disagrees at four depths), and bitcoin trust
(pmax\,$=\,0.610$, disagrees at four depths).

The freshweb~(Stony) case warrants specific comment.  Its reference
classification at $n = 20$ is Reactive with a posterior maximum of
exactly $0.500$ (the minimum possible value before the label flips)
and the reference posterior for the alternative class (Stable) at
each shorter depth is between $0.285$ and $0.495$.  This network lies at
the decision boundary by definition; its apparent instability is therefore
a reflection of its intrinsic ambiguity rather than a sensitivity of the
classification procedure.

Darwin and bitcoin trust are classified as Diffuse at depths $n = 5$--$14$
but converge to Stable at depth~$17$ and the reference depth.  Their
deep-memory signatures differentiate from the Diffuse cluster only when
sufficient cascade depth is included in the feature representation,
indicating that their Stable membership is a property of their deep
rather than shallow cascade structure.

\paragraph{GMM diagnostics.}
The BIC-optimal number of components is $K = 4$ at depths $11$, $17$,
and $20$, matching the reference, and $K = 5$ at depths $5$, $8$, and $14$
(Fig.~\ref{fig:depth_robustness}d).  At $K = 5$, the additional component
absorbs borderline networks whose deep-memory signatures have not yet
differentiated at the truncated cascade depth; as more cascade information
is included, BIC penalises the extra parameter and collapses back to
$K = 4$.  The bootstrap ARI (measuring internal GMM stability across
resamplings at each depth) is comparable across all depths (range
$0.438$--$0.590$, reference depth $0.459 \pm 0.174$), confirming that
the GMM partition is similarly well-defined at all tested depths and that
the improved agreement with the reference at deeper cascade levels reflects
convergence of the underlying feature representations, not an increase in
cluster separability.

\paragraph{Summary.}
The classification is fully robust for the $22$ of $34$ networks
($65\%$) classified with high confidence at the reference depth:
these networks receive identical regime and driver labels at all tested
depths from $n = 5$ onwards.  Instability is confined to the $12$
lower-confidence networks, all with pmax\,$\leq\,0.61$, and is most
pronounced in the three genuinely borderline cases (pmax\,$\leq\,0.54$)
whose classification at the reference depth is itself uncertain.
The driver axis is robust at all depths, reaching full agreement by $n
= 14$.  The reference depth $n = 20$ was selected as the finest cascade
level available across all corpus networks; the depth robustness analysis
confirms that the qualitative conclusions of the main analysis are not
sensitive to this choice.

\section{Binary comparison for the Chesapeake Bay wet-season food web}
\label{sec:baywet_binary}

\paragraph{Motivation.}
The Chesapeake Bay wet-season food web~\cite{Baird1989} is the one network
in the corpus for which the GFM deviation is negative
($\Delta_{\mathrm{GFM}} = -0.004$, SNR $= 0.44$).  The main text
attributes this to the measurement resolution of the model-estimated
carbon flows: the edge weights are seasonal mass-balance estimates rather
than directly measured interaction strengths, and their precision may be
insufficient to encode the functional memory structure that the trophic
topology itself carries.  To test this attribution directly, we repeated
the full analysis pipeline on a binarised version of the same network
(denoted \textit{baywet-binary}), in which all edge weights are replaced
by unit values.  This isolates the contribution of the adjacency structure
from that of the carbon flow magnitudes, and asks whether the topology
alone satisfies the GFM law.

\paragraph{GFM deviation: sign reversal and mechanism.}
The binarised network produces a strongly positive GFM deviation
($\Delta_{\mathrm{GFM}} = +0.022$, SNR $= 11.9$,
CI $= [+0.022, +0.023]$), compared to the near-zero negative value in the
weighted network ($\Delta_{\mathrm{GFM}} = -0.004$, SNR $= 0.44$,
CI $= [-0.005, -0.002]$).  The reversal is driven by a large shift in the
GFM null ensemble, not in the real network: the real NMR changes by only
$-0.002$ upon binarisation ($0.097 \to 0.096$), while the GFM null mean
falls by $0.028$ ($0.101 \to 0.073$).  The mechanism is the following.
Model-estimated carbon flows are heterogeneous but not coherently organised
along deep directed paths: their variability is sufficiently large that
random weight permutations can occasionally produce higher walk-product
amplification than the original assignment, inflating the GFM null above
the real NMR.  Unit weights carry no such incidental heterogeneity; every
weight permutation of a binary network is equivalent, so the GFM null
measures the pure topological contribution of the adjacency structure.
The topology of the Chesapeake Bay food web does encode coherent geometric
memory structure: the binary GFM deviation of $+0.022$ is comparable in
magnitude to the five other food webs in the corpus, whose GFM deviations
range from $+0.003$ (baydry) to $+0.039$ (freshweb Stony).

\paragraph{DCM deviation.}
The DCM deviation is positive in both versions but changes substantially
upon binarisation ($\Delta_{\mathrm{DCM}}$: $+0.003$, SNR $0.25$ $\to$
$+0.012$, SNR $9.2$).  In the weighted network the mesoscale contribution
is marginally significant (CI $= [+0.001, +0.006]$); in the binary
network it is strongly significant (CI $= [+0.011, +0.012]$).  This
confirms that higher-order trophic wiring organisation contributes
positively to memory depth in the Chesapeake Bay food web independent
of the flow magnitudes, as it does in all six food webs.

\paragraph{Structural driver.}
The driver label changes from Meso (weighted) to Geo-M (binary).  In the
weighted network, the Wasserstein fractions are
$(\phi^{\mathrm{DCM}}, \phi^{\mathrm{GFM}}, \phi^{\mathrm{PNM}}) =
(0.67, 0.21, 0.12)$: DCM dominates because the noisy carbon flows
generate little distinctive GFM signal relative to the mesoscale wiring
structure.  In the binary network the fractions invert to $(0.18, 0.64,
0.19)$: GFM dominates, placing the Chesapeake Bay food web in the
geometry-dominated sector alongside the other freshwater and
palaeontological food webs in the corpus (litte rock, messal shale,
freshweb Stony).  The driver assignment for the weighted baywet therefore
reflects the information content of the flow estimates rather than the
functional organisation of the trophic network.

\paragraph{Polarity response.}
The PNM slope reverses sign upon binarisation: $s = +0.004$ (weighted,
near-zero) $\to$ $s = -0.018$ (binary, clearly negative).  The positive
slope of the weighted network made it the sole outlier to the food-web
negative-slope pattern reported in the main text; the binary network
restores consistency with that pattern.  Directionality sustains deep
memory in the Chesapeake Bay food web topology, as in every other food
web in the corpus.

\paragraph{Dynamical organisation.}
The regime label is Reactive in both versions (pmax $0.92$ weighted,
$0.95$ binary).  The deep-memory dynamical class is topologically
determined and insensitive to whether model-estimated or unit weights are
used, consistent with the observation that NMR itself changes by less than
$2\%$ upon binarisation.

\paragraph{Summary.}
Table~\ref{tab:baywet_binary} reports the full comparison.  Every
anomalous feature of the weighted Chesapeake Bay wet-season food web
--- the negative GFM deviation, the Meso driver, the near-zero positive
polarity slope --- disappears when the model-estimated carbon flows are
replaced by unit weights.  The topology of the food web satisfies the GFM
law strongly, carries a clear mesoscale signal, and exhibits the negative
polarity slope characteristic of trophic systems.  The weighted network's
anomalies are therefore artefacts of the resolution of the flow estimates,
not properties of the Chesapeake Bay ecosystem.  This confirms the
interpretive framework of the GFM comparison as a plausibility test for
edge weight quality: a near-zero or negative deviation is a diagnostic
signal that the weight configuration may not faithfully encode the
functionally relevant interaction geometry.

\begin{table}[h]
\centering
\caption{Comparison of the Chesapeake Bay wet-season food web with and
without model-estimated carbon flow weights.  SNR $= |\Delta|/\sigma$
where $\sigma$ is the bootstrap standard deviation.\\
}
\label{tab:baywet_binary}
\small
\begin{tabular}{lrr}
\toprule
Observable & Weighted & Binary \\
\midrule
NMR                          & $0.097$  & $0.096$  \\
NMR$_{\mathrm{GFM}}$ (null) & $0.101$  & $0.073$  \\
$\Delta_{\mathrm{GFM}}$      & $-0.004$ & $+0.022$ \\
GFM SNR                      & $0.44$   & $11.9$   \\
GFM 95\% CI                  & $[-0.005,\,-0.002]$ & $[+0.022,\,+0.023]$ \\
\midrule
$\Delta_{\mathrm{DCM}}$      & $+0.003$ & $+0.012$ \\
DCM SNR                      & $0.25$   & $9.2$    \\
DCM 95\% CI                  & $[+0.001,\,+0.006]$ & $[+0.011,\,+0.012]$ \\
\midrule
PNM slope $s$                & $+0.004$ & $-0.018$ \\
$W_1^{\mathrm{DCM}}$         & $0.083$  & $0.020$  \\
$W_1^{\mathrm{GFM}}$         & $0.026$  & $0.072$  \\
$\phi^{\mathrm{DCM}}$        & $0.671$  & $0.177$  \\
$\phi^{\mathrm{GFM}}$        & $0.213$  & $0.637$  \\
\midrule
Driver                       & Meso     & Geo-M    \\
Regime                       & Reactive & Reactive \\
Regime $p_{\max}$            & $0.92$   & $0.95$   \\
\bottomrule
\end{tabular}
\end{table}

\section{Falsification: GFM comparison under random graph null models}
\label{sec:random_controls}
\paragraph{Design.}
The GFM comparison is a controlled intervention: it asks whether observed
interaction strengths organise functional memory more deeply than weight
fluctuations on the same binary topology. A positive deviation
$\Delta_{\mathrm{GFM}}>0$, with
$\Delta_{\mathrm{GFM}}={\rm NMR}_{\rm real}-{\rm NMR}_{\rm GFM}$,
is therefore informative only if the same pipeline does not produce a
systematic positive signal when edge weights have no coherent geometric
organisation by construction. To test this, we applied the full analysis
pipeline to two families of synthetic directed networks with independently
assigned positive edge weights. The synthetic weights were iid lognormal
and normalised to unit mean; the GFM ensemble then perturbed these weights
on the fixed binary support using the same continuous-weight GFM procedure
as in the empirical analysis. Thirty-four graphs were generated from each
family, all with $n=200$ nodes, using the same convergence tolerance
($\varepsilon=10^{-5}$) and $n_{\mathrm{null}}=100$ GFM resamples.

The two families are:
\begin{itemize}
\item \textbf{ER-iid}: sparse directed Erdős--Rényi controls generated with
Bernoulli parameter $p=0.04$ before removal of self-loops and collapse to
continuous weighted support (mean final density $0.0346$, mean $1376$
directed edges).
\item \textbf{WS-iid}: directed Watts--Strogatz controls with $k=6$ outgoing
ring neighbours and rewiring probability $p_{\mathrm{rew}}=0.4$, with no
parallel edges retained in the final weighted support (mean final density
$0.0299$, mean $1189$ directed edges).
\end{itemize}

\paragraph{Results.}
For ER-iid, $\Delta_{\mathrm{GFM}}<0$ in all thirty-four graphs, with
confidence intervals entirely below zero in every case
(mean $\Delta_{\mathrm{GFM}}=-0.0032$, range
$[-0.0060,-0.0013]$; violation rate $1.00$, 95\% CI
$[0.90,1.00]$; Fig.~\ref{fig:random_controls}). Thus, in sparse random
directed graphs with iid weights, the GFM ensemble systematically exhibits
larger non-Markovianity than the original synthetic graph
(mean ${\rm NMR}_{\rm GFM}/{\rm NMR}_{\rm real}=1.025$).

For WS-iid, the result was equally strong. Despite the small-world support,
$\Delta_{\mathrm{GFM}}<0$ in all thirty-four graphs, again with confidence
intervals entirely below zero in every case
(mean $\Delta_{\mathrm{GFM}}=-0.0028$, range
$[-0.0046,-0.0013]$; violation rate $1.00$, 95\% CI
$[0.90,1.00]$; Fig.~\ref{fig:random_controls}). The GFM null ensemble
again outperformed the iid-weighted network
(mean ${\rm NMR}_{\rm GFM}/{\rm NMR}_{\rm real}=1.019$).

\paragraph{Interpretation.}
Both random model families violate the empirical GFM law, despite having
different topological organisation. These negative controls show that a
positive $\Delta_{\mathrm{GFM}}$ is not a generic consequence of the EMD
pipeline, graph sparsity, positive edge weights, or the GFM resampling
procedure itself. They also show that small-world topology alone is
insufficient: without coherent placement of interaction strengths on the
directed support, the empirical memory-deepening effect disappears and is
reversed. The systematic positive $\Delta_{\mathrm{GFM}}$ observed in the
empirical corpus is therefore a signature of organised weight geometry,
rather than a generic property of weighted directed graphs.
\begin{figure*}[!h]
  \centering
  \includegraphics[width=\textwidth]{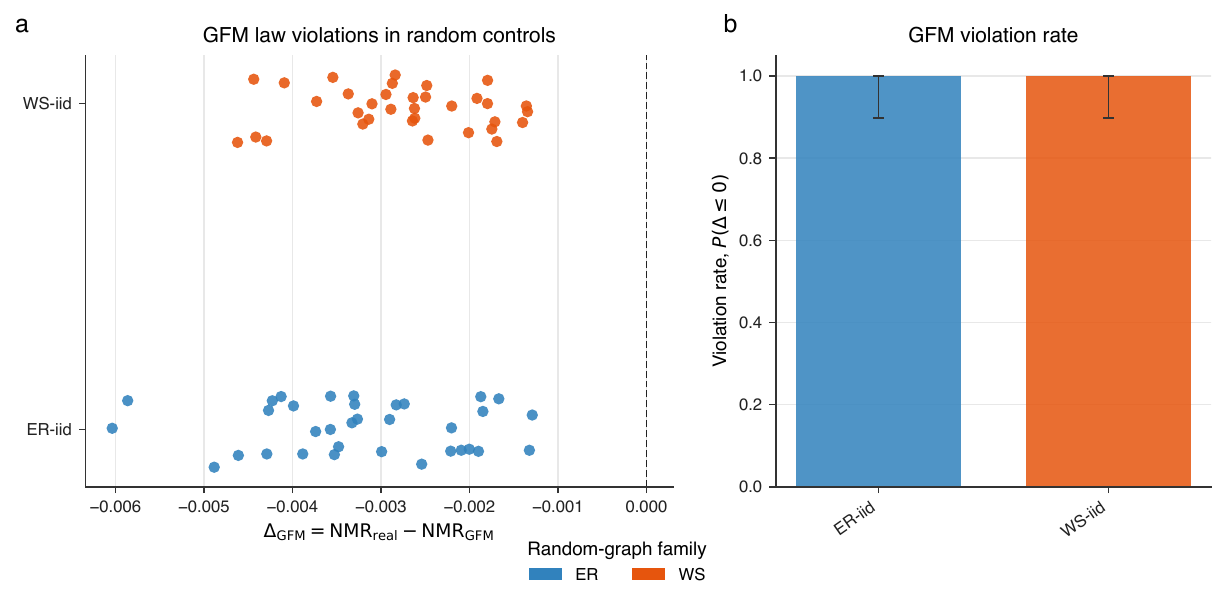}
 \caption{
\textbf{GFM deviations in synthetic random-graph controls with independently distributed weights.}
\textbf{a}, Geometry-fluctuation shift
$\Delta_{\mathrm{GFM}}=\NMR_{\mathrm{real}}-\NMR_{\mathrm{GFM}}$
for thirty-four Erdős--Rényi controls (ER-iid, blue) and thirty-four
Watts--Strogatz controls (WS-iid, orange), each with $n=200$ nodes.
Synthetic edge weights were drawn independently from a positive lognormal
distribution and normalised to unit mean; the GFM ensemble was then generated
on the same binary support using the same continuous-weight resampling
procedure as in the empirical analysis. Each point is the mean over
$n_{\mathrm{null}}=100$ GFM realisations. The dashed vertical line marks
$\Delta_{\mathrm{GFM}}=0$; values to the left indicate that the GFM ensemble
has larger non-Markovianity than the original iid-weighted graph. Both ER-iid
and WS-iid controls fall entirely to the left of zero.
\textbf{b}, Violation rate $P(\Delta_{\mathrm{GFM}}\leq 0)$ for each group.
Error bars are Wilson 95\% confidence intervals. The violation rate is
$1.00$ for both ER-iid and WS-iid controls (95\% CI $[0.90,1.00]$ in each
case), in contrast to the empirical corpus, where only $2$ of $34$ networks
violate the GFM law. Thus, random iid-weighted graphs do not reproduce the
systematic positive GFM shift observed in real functionally weighted networks.
}
  \label{fig:random_controls}
\end{figure*}

\section{Full functional memory classification}

\begin{table}[!ht]
\centering
\scriptsize
\setlength{\tabcolsep}{3.5pt}
\renewcommand{\arraystretch}{0.96}
\caption{\textbf{Structural characteristics and functional-memory 
classification of the 34 empirical networks.}
The driver label identifies the dominant structural source of 
deviation from the null ensembles, whereas the memory regime 
describes the organisation of hierarchical functional memory.}
\label{tab:network_driver_regime_classification}
\begin{tabular}{llccll} 
\toprule
\textbf{Network} &
\textbf{Domain} &
{\textbf{Nodes}} &
{\textbf{Edges}} &
\textbf{Driver} &
\textbf{Memory regime} \\
\midrule

\textit{E. coli} transcription~\cite{ShenOrr2002} &
Gene regulatory & 423 & 578 & Geo-M & Diffuse \\

Figeys PPI~\cite{Ewing2007} &
Protein interaction & 2239 & 6452 & Geo-M & Diffuse \\

Stelzl PPI~\cite{Stelzl2005} &
Protein interaction & 1706 & 6207 & Geo-H & Diffuse \\

Plant photosystem I~\cite{Montepietra2020} &
Biophysical & 192 & 29390 & Geo-L & Pulse \\

TRRUST (human)~\cite{Han2015} &
Gene regulatory & 2862 & 9396 & Geo-H & Diffuse \\

TRRUST (mouse)~\cite{Han2015} &
Gene regulatory & 2456 & 7057 & Geo-H & Diffuse \\

Yeast transcription~\cite{Milo2002} &
Gene regulatory &  916 & 1094 & Geo-H & Diffuse \\

\addlinespace[2pt]
\textit{C. elegans} soma~\cite{Moutuou2025b} &
Neural & 280 & 12071 & Meso & Reactive \\

\textit{C. elegans} chemical synapses (hermaphrodite)~\cite{Cook2019} &
Neural connectome & 443 & 4207 & Meso & Reactive \\

\textit{C. elegans} chemical synapses (male)~\cite{Cook2019} &
Neural  connectome& 905 & 5300 & Meso & Stable \\

\textit{C. intestinalis} &
Neural connectome & 205 & 2903 & Geo-H & Reactive \\

\textit{Drosophila} larva~\cite{Winding2023} &
Neural connectome & 2952 & 110677 & Meso & Reactive \\

\addlinespace[2pt]
DBLP citations~\cite{nr} &
Citation & 12591 & 49743 & Geo-M & Diffuse \\

\addlinespace[2pt]
Bay Dry food web~\cite{nr} &
Food web & 128 & 2137 & Meso & Reactive \\

Bay Wet food web~\cite{konect,Peixoto2020} &
Food web & 128 & 2106 & Meso & Reactive \\

Mangrove Wet food web~\cite{nr} &
Food web & 97 & 1492 & Meso & Reactive \\

Little Rock food web~\cite{konect,Peixoto2020} &
Food web & 183 & 2494 & Geo-L & Reactive \\

Messel Shale food web~\cite{Dunne2014,Peixoto2020} &
Food web & 700 & 6444 & Geo-M & Stable \\

Stony Stream food web~\cite{Thompson2003,Peixoto2020} &
Food web & 112 & 832 & Geo-L & Reactive \\

\addlinespace[2pt]
U.S. faculty hiring~\cite{Wapman2022,Peixoto2020} &
Transportation & 3284 & 61936 & Meso & Stable \\

\addlinespace[2pt]
Darwin word adjacency~\cite{Milo2004,Peixoto2020} &
Information & 7381 & 46281 & Geo-Meso & Stable \\

Japanese word adjacency~\cite{Milo2004,Peixoto2020} &
Information & 2704 & 8300 & Geo-L & Diffuse \\

Political blogs~\cite{Adamic2005,Peixoto2020} &
Information & 1224 & 19025 & Geo-M & Stable \\

\addlinespace[2pt]
Gnutella P2P~\cite{snapnets} &
Social & 10876 & 39994 & Geo-H & Diffuse \\

Wikipedia voting~\cite{snapnets} &
Social & 7115 & 103689 & Geo-H & Stable \\

\addlinespace[2pt]
Congress Twitter~\cite{snapnets} &
Social & 475 & 13289 & Geo-M & Reactive \\

Fediverse~\cite{Peixoto2020} &
Social & 4860 & 484164 & Geo-L & Reactive \\

Spanish high school~\cite{RuizGarcia2023} &
Social & 409 & 8557 & Geo-M & Reactive \\

Facebook-like social~\cite{Opsahl2009} &
Social & 1899 & 59835 & Meso & Stable \\

ANU residence hall~\cite{Freeman1998,Peixoto2020} &
Social & 217 & 2672 & Geo-L & Reactive \\

Bitcoin trust~\cite{Kumar2016,snapnets} &
Social & 5881 & 35592 & Dir & Stable \\

\addlinespace[2pt]
JDK dependencies~\cite{konect} &
Software & 6434 & 150985 & Dir & Stable \\

\addlinespace[2pt]
OpenFlights~\cite{Peixoto2020} &
Transportation & 3214 & 36907 & Geo-L & Stable \\

U.S. air traffic~\cite{Peixoto2020} &
Transportation & 2278 & 6390340 & Meso & Pulse \\

\bottomrule
\end{tabular}
\end{table}

\end{document}